\documentclass[11pt]{article}

\usepackage{authblk}
\usepackage{amsmath}
\usepackage{amsthm}
\usepackage{amssymb}
\usepackage{amsfonts}
\usepackage{latexsym}
\usepackage{float}
\usepackage{graphicx}
\usepackage{subfigure}
\usepackage{enumerate}
\usepackage{enumitem}
\usepackage{mathtools}
\usepackage{afterpage}
\usepackage{url}
\usepackage[breaklinks]{hyperref}
\usepackage{mathptm}
\usepackage{charter}
\usepackage{fullpage}
\usepackage{ifthen}
\usepackage[margin=1in]{geometry}

\usepackage{xy}
\xyoption{matrix}
\xyoption{frame}
\xyoption{arrow}
\xyoption{arc}

\newtheorem{definition}{Definition}
\newtheorem{theorem}{Theorem}

\newtheorem{lemma}{Lemma}
\newtheorem{claim}{Claim}
\newtheorem{corollary}{Corollary}

\floatstyle{ruled}
\newfloat{algorithm}{htb}{loa}
\floatname{algorithm}{Algorithm}

\newcommand{\beq}{\begin{equation}}
\newcommand{\eeq}{\end{equation}}
\newcommand{\bea}{\begin{eqnarray}}
\newcommand{\eea}{\end{eqnarray}}
\newcommand{\ba}{\begin{array}}
\newcommand{\ea}{\end{array}}
\newcommand{\beba}{\begin{equation}\begin{array}{lll}}
\newcommand{\eeea}{\end{array}\end{equation}}
\newcommand{\bc}{\begin{cases}}
\newcommand{\ec}{\end{cases}}
\newcommand{\bpm}{\begin{pmatrix}}
\newcommand{\epm}{\end{pmatrix}}
\newcommand{\bit}{\begin{itemize}}
\newcommand{\eit}{\end{itemize}}
\newcommand{\ben}{\begin{enumerate}}
\newcommand{\een}{\end{enumerate}}

\newcommand{\R}{\mathbb{R}}

\newcommand{\N}{\mathbb{N}}

\newcommand{\deq}{\triangleq}

\newcommand{\lb}{\left(}
\newcommand{\rb}{\right)}
\newcommand{\lr}[1]{\lb #1 \rb}

\renewcommand{\sl}{\sum\limits}

\renewcommand{\span}[1]{\mathrm{span} \lb \{ #1  \} \rb}
\renewcommand{\ker}[1]{\mathrm{Ker} \lb #1 \rb}

\renewcommand{\Col}[1]{\mathcal{C} \lb #1 \rb}

\newcommand{\mmax}[1]{\mathrm{max} \lb #1 \rb}

\newcommand{\poly}[1]{\mathrm{poly}\lb #1 \rb}
\newcommand{\polylog}[1]{\mathrm{polylog}\lb #1 \rb}
\newcommand{\bgo}[1]{O \lb #1 \rb}
\newcommand{\omg}[1]{\Omega\lb #1 \rb}
\newcommand{\tht}[1]{\Theta\lb #1 \rb}

\newcommand{\tbgo}[1]{\tilde{O} \lb #1 \rb}

\newcommand{\tp}{\otimes}

\newcommand{\la}{\leftarrow}

\newcommand{\refl}[1]{\mathrm{Ref}_{#1}}

\newcommand{\ket}[1]{\left|#1\right\rangle}
\newcommand{\bra}[1]{\left\langle #1\right|}
\newcommand{\braket}[2]{\left\langle #1|#2 \right\rangle}
\newcommand{\ketbra}[2]{\left|#1\rangle \langle #2  \right |}

\newcommand{\zo}{\{0,1\}}
\newcommand{\zon}{\{0,1\}^n}

\newcommand{\D}{\mathcal{D}}
\renewcommand{\H}{\mathcal{H}}
\renewcommand{\P}{\mathcal{P}}
\newcommand{\Q}{\mathcal{Q}}

\newcommand{\ws}{\mathrm{wsize}}
\newcommand{\ecs}{\mathcal{E}}

\title{Span-program-based quantum algorithm for tree detection}

\author{Guoming Wang \\{Computer Science Division, University of California, Berkeley, U.S.A.} {\vspace{-5ex}} \thanks{Email: wgmcreate@berkeley.edu.}}

\date{}

\begin{document}

\maketitle

\begin{abstract}
Span program is a linear-algebraic model of computation originally proposed for studying the complexity theory. Recently, it has become a useful tool for designing quantum algorithms. In this paper, we present a time-efficient span-program-based quantum algorithm for the following problem. Let $T$ be an arbitrary tree. Given query access to the adjacency matrix of a graph $G$ with $n$ vertices, we need to determine whether $G$ contains $T$ as a subgraph, or $G$ does not contain $T$ as a minor, under the promise that one of these cases holds. We call this problem the subgraph/not-a-minor problem for $T$. We show that this problem can be solved by a bounded-error quantum algorithm with $O(n)$ query complexity and $\tilde{O}(n)$ time complexity. The query complexity is optimal, and the time complexity is tight up to polylog factors. 
\end{abstract}

\section{Introduction}
\label{sec:intro}

Given two graphs $G$ and $P$, it is natural to ask whether $G$ contains $P$ as a subgraph. This problem, known as the \emph{subgraph isomorphism} problem, has numerous applications in cheminformatics \cite{cheminf}, circuit design \cite{circuit} and software engineering \cite{se}. If $G$ and $P$ are both given as input, this problem is NP-complete, and hence it is unlikely to be solvable in polynomial time. However, if $P$ is fixed and only $G$ is given as input, then this problem, usually called the \emph{$P$-containment} problem, can be solved efficiently. Specifically, if $P$ contains $k$ vertices, then the $P$-containment problem can be solved in $\bgo{n^k}$ classical time, where $n$ is the number of vertices in $G$. In fact, by exploiting $P$'s structure cleverly, we can usually do much better. For example, if $P$ is a tree, then the $P$-containment problem can be solved in $O(n^2)$ classical time \cite{colorcoding}  (assuming $G$ is given by the $n \times n$ adjacency matrix). 

Recently, there has been rising interest in developing fast quantum algorithms for the subgraph containment problem. In particular, the problem of triangle finding has received the most attention, perhaps due to its simplicity and its application to boolean matrix multiplication. Magniez, Santha and Szegedy \cite{triangle} first gave two quantum algorithms for this problem, one based on Grover's search and the other based on quantum walks \cite{oldqw}. They achieved $\tbgo{n^{13/10}}$ quantum query complexity for this problem. Then, Belovs \cite{lg} used learning graphs to show that this problem has $\bgo{n^{35/27}}$ quantum query complexity. Subsequently, this result was improved to $\bgo{n^{9/7}}$, first by Lee, Magniez and Santha \cite{lgtriangle} who used learning graphs, second by Jeffery, Kothari and Magniez who used \emph{nested} quantum walks \cite{nqw}. 

There has also been work on quantum algorithms for detecting other patterns. Childs and Kothari \cite{minorclosed} studied the quantum query complexity of detecting paths, claws, cycles and bipartite patterns, etc. Later, Belovs and Reichardt \cite{spstconnect} showed that detecting paths and subdivided stars can be done in $O(n)$ queries and $\tbgo{n}$ time. Finally, there were also upper bounds on the quantum query complexity of detecting arbitrary patterns  \cite{lgsubgraph1,lgsubgraph2,nqw}. It is worth mentioning that most of the above algorithms are query-efficient but not time-efficient.

In this paper, we present a \emph{time-efficient} quantum algorithm for the following variant of tree containment problem. Let $T=(V_T,E_T)$ be an arbitrary tree. Given query access to the adjacency matrix of a graph $G$ with $n$ vertices, we need to determine whether $G$ contains $T$ as a subgraph, or $G$ does not contain $T$ as a \emph{minor}, under the promise that one of these cases holds. We call this problem the \emph{subgraph/not-a-minor} problem for $T$. We show that this problem can be solved by a bounded-error quantum algorithm with $\bgo{n}$ query complexity and $\tbgo{n}$ time complexity \footnote{We use the symbol $\tilde{O}$ to suppress polylog factors. Namely, $\tbgo{f(n)}=\bgo{f(n)\lr{\log{f(n)}}^{b}}$ for some constant $b \ge 0$.}. Formally,

\begin{theorem}
Let $T=(V_T,E_T)$ be an arbitrary tree. Then the subgraph/not-a-minor problem for $T$ can be solved by a bounded-error quantum algorithm with $O(n)$ query complexity and $\tilde{O}(n)$ time complexity.
\end{theorem}

Note that by a reduction from the unstructured search, one can show that the subgraph/not-a-minor problem for any $T$ has $\omg{n}$ quantum query complexity (see Proposition 4 of \cite{spstconnect}). Thus, the query complexity of our algorithm is optimal, and its time complexity is tight up to polylog factors. 

Our algorithm is developed based on \emph{span programs}. Span program \cite{sp} is a linear-algebraic model of computation originally proposed for the study of complexity theory. Informally speaking, a span program is composed of a \emph{target} vector and a collection of \emph{input} vectors from an inner-product space. It accepts an input if and only if the target vector lies in the span of the \emph{available} input vectors on this input. Its complexity is captured by a measure called \emph{witness size}. In a recent breakthrough, Reichardt \cite{spqqc1,spqqc2} proved that for any (partial) boolean function, the optimal witness size of span programs for this function is within a constant factor of the bounded-error quantum query complexity for the same function. His result leads to a novel approach to develop quantum algorithms based on span programs. To date, several quantum algorithms \cite{sprank,spstconnect,spgc,spgametree,spformula1,spformula2} have been designed in this way. In fact, learning graphs are a special class of span programs, and hence the learning-graph-based algorithms \cite{kdist1,kdist2,lg,lgtriangle,lgsubgraph1,lgsubgraph2} can be also viewed as span-program-based. Nevertheless, most of these algorithms are only query-efficient but not time-efficient. The efficient evaluation of span programs and learning graphs can be challenging.

Our work is closely related to the span program for undirected $st$-connectivity \cite{spstconnect}, which works as follows. In order to test whether $s$ and $t$ are connected in an undirected graph $G=(V,E)$, we build a span program with target vector $\ket{s}-\ket{t}$ and input vectors $\ket{u}-\ket{v}$ for any $u,v \in V$. The input vector $\ket{u}-\ket{v}$ is available if and only if $(u,v) \in E$. Then it is obvious that the target vector lies in the span of the available input vectors if and only if there is a path connecting $s$ and $t$ in $G$. In other words, the idea of this span program is that we try to run a single flow from $s$ to $t$ in $G$. 

We utilize the span program for undirected $st$-connectivity in the task of tree detection as follows. Suppose $T$ has root $r$ and leaves $f_1$, $f_2$, $\dots$, $f_k$. Then, in order to test whether $G$ contains $T$ as a subgraph, we verify that $G$ contains $k$ paths such that: (1) the $j$-th path resembles the path from $r$ to $f_j$ in $T$; (2) these $k$ paths overlap in certain way so that their union looks like $T$. To accomplish this, we use a technique called \emph{``parallel flows"}. Namely, we run $k$ flows in parallel such that the $j$-th flow corresponds to the path from $r$ to $f_j$, and we let these flows interfere somehow to ensure that they respect the correlations among the $k$ paths. Algebraically, our span program is the ``direct sum" of $k$ span programs for undirected $st$-connectivity, but these $k$ span-programs are also correlated somehow so that their solutions (i.e. the $st$-paths) overlap in certain way. This parallel-flow technique might be useful somewhere else (see Section \ref{sec:con}).

For analyzing the witness size of our span program, we prove a theorem (i.e. Claim \ref{clm:structure}) about the structure of graphs that do not contain a given tree as a minor, which might be of independent interest.

\section{Preliminaries}
\label{sec:prelim}

\subsection{Notation}

Let $[n]=\{1,2,\dots,n\}$. For any matrix $A$, let $\Col{A}$ be the column space of $A$, and let $\refl{A}$ be the reflection about $\Col{A}$. Furthermore, let $\ker{A}$ be the kernel of $A$.

For any tree $T=(V_T,E_T)$, let $V_{T,i}$ and $V_{T,l}$ be the set of internal nodes and leaves of $T$ respectively. Also, let $r$ be denote the root of $T$. For any $x \in V_T$, let $C(x)$ be the set of $x$'s children, and let $T^x$ be the subtree of $T$ rooted at $x$, and let $V^x_T$ be the set of nodes in $T^x$, and let $V^x_{T,i}$, $V^x_{T,l}$ be the set of internal and leaf nodes in $T^x$ respectively. Furthermore, we will add two special nodes $s$ and $t$, and set $V^s_{T,l}=V^t_{T,l}=V_{T,l}$. Then, for any $x,y \in V_T \cup \{s,t\}$, let $S_{x,y} = V^x_{T,l} \cap V^y_{T,l}$. We say $x,y \in V_T \cup \{s,t\}$ are \emph{adjacent} (denoted by $x \sim y$) if $(x,y)\in E_T$ or $\{x,y\}=\{s,t\}$ or $\{x,y\}=\{s,r\}$ or $\{x,y\}=\{t,f\}$ for some $f \in V_{T,l}$. Then for any $x \in V_T \cup \{s,t\}$ and $f \in V^x_{T,l}$, let $M(x,f)=\{y \in V_{T}:~x\sim y,~f\in S_{x,y}\}$. It is easy to check that $|M(x,f)|=2$.

For any graph $G=(V_G,E_G)$ and $U \subseteq V_G$, let $G|_{U} \deq (U,E_{G|_U})$ be the \emph{induced subgraph} of $G$ on the vertex set $U$. Namely, for any $u,v \in U$, $(u,v) \in E_{G|_U}$ if and only if $(u,v) \in E_G$.

\subsection{Graph theory}

For any two graphs $G=(V_G,E_G)$ and $H=(V_H,E_H)$, we say that $H$ is a \emph{minor} of $G$ if $H$ can be obtained from $G$ by deleting and contracting edges of $G$, and removing isolated vertices. Here, contracting an edge $(u,v)$ means replacing $u$ and $v$ by a new vertex and connecting this vertex to the original neighbors of $u$ and $v$. 

Suppose $\ecs$ is an arbitrary subset of $E_G$. Then we use $\ecs(G)$ to denote the graph obtained by contracting the edges in $\ecs$. Note that $\ecs(G)$ is well-defined, because the final graph is independent of the order of the edge contractions. Moreover, if the vertices $v_{1}, \dots, v_{k} \in V_G$ are combined together in $\ecs(G)$ (and no other vertex is combined with them), then we denote this new vertex as $w \deq \{v_{1},\dots,v_{k}\}$ and we say that $w$ {\emph {contains}} $v_{1}$, $\dots,$ $v_{k}$. Finally, for any $u \in V_G$ we say that $u$ is {\emph {involved}} in $\ecs$ if there exists $v \in V_G$ such that $(u,v) \in \ecs$.

\subsection{Span program and quantum query complexity}

Span program is a linear-algebraic model of computation defined as follows:

\begin{definition}[Span program \cite{sp}]
A span program $\P$ is a $6$-tuple $(n,d,\ket{\tau},\{\ket{v_j}: j \in [m] \},I_{{free}},\{I_{i,b}:i \in [n], b\in \zo\})$, where $\ket{\tau} \in \R^d$,  $\ket{v_j} \in \R^d$ for any $j \in [m]$, and $I_{free} \cup \lb \cup_{i=1}^n I_{i,x_i} \rb = I \deq [m]$. $\ket{\tau}$ is called is the {target} vector, and each $\ket{v_j}$ is called an input vector. For any $j \in I_{free}$, we say that $\ket{v_j}$ is a {free} input vector; for any $j \in I_{i,b}$ for some $i \in [n]$ and $b \in \zo$, we say that $\ket{v_j}$ is labeled by $(i,b)$.
 
To $\P$ corresponds a boolean function $f_{\P}:\zon \to \zo$ defined as follows: for $x=x_1\dots x_n \in \zon$,
\beba
f_{\P}(x)=
\bc
1, & ~~~\mathrm{if~}\tau \in \span{\ket{v_j}: j \in I_{free} \cup \lb \cup_{i=1}^n I_{i,x_i} \rb } , \\
0, & ~~~\mathrm{otherwise}.
\ec 
\eeea
Namely, on input $x$, only the $\ket{v_j}$'s with $j \in I_{free} \cup \lb \cup_{i=1}^n I_{i,x_i} \rb$ are {available}, and $f_P(x)=1$ if and only if the target vector lies in the span of the available input vectors.
\end{definition}
For convenience, we say that $\P$ \emph{accepts} or \emph{rejects} $x$ if $f_{\P}(x)=1$ or $0$, respectively.

The complexity of a span program is measured by its \emph{witness size} defined as follows:

\begin{definition}[Witness size \cite{spqqc1}]
Let $\P=(n,d,\ket{\tau},\{\ket{v_j}: j \in [m]\},I_{free},\{I_{i,b}:i \in [n], b\in \zo\})$ be a span program. Let $I=I_{free} \cup \lb \cup_{i=1}^n \cup_{b \in \zo} I_{i,b} \rb$ and let
$A=\sum_{j \in I} \ketbra{v_j}{j}$. Then, for any $x \in \zon$, let 
$I(x)=I_{free} \cup \lb \cup_{j=1}^n I_{j,x_j} \rb$, 
$\bar{I}(x)=\cup_{j=1}^n I_{j,x_j}.$
Then, let 
$\Pi(x)=\sum_{j \in I(x)} \ketbra{j}{j}$, 
$\bar{\Pi}(x)=\sum_{j \in \bar{I}(x)} \ketbra{j}{j}$.
The witness size of $\P$ on $x$, denoted by $\ws(\P,x)$, is defined as follows:
\bit
\item If $f_P(x)=1$, then $\ket{\tau} \in \Col{A(\Pi(x))}$, so there exists $\ket{w} \in \R^m$ satisfying $A \Pi(x) \ket{w} = \ket{\tau}$.
Any such $\ket{w}$ is a (positive) witness for $x$, and its size is defined as $\|\bar{\Pi}(x)\ket{w}\|^2$. Then $\ws(P,x)$ is defined as the minimal size among all such witnesses.
\item If $f_P(x)=0$, then $\ket{\tau} \not\in \Col{A(\Pi(x))}$, so there exists $\ket{w'} \in \R^d$ satisfying 
$\braket{\tau}{w'}=1$ and $\Pi(x)A^{\dagger}\ket{w'}=0$.
Any such $\ket{w'}$ is a (negative) witness for $x$, and its size is defined as $\|A^{\dagger}\ket{w'}\|^2$. Then $\ws(P,x)$ is defined as the minimal size among all such witnesses.
\eit
\end{definition}

For any $\D \subseteq \zon$ and $b \in \zo$, let
\beq
\ws_b(\P,\D)=\max\limits_{x\in \D: f_P(x)=b} \ws(P,x).
\eeq
Then the witness size of $\P$ over domain $\D$ is defined as
\beq
\ws(\P,\D)=\sqrt{\ws_0(P,\D)\ws_1(P,\D)}.
\eeq

Surprisingly, for any (partial) boolean function, the optimal witness size of span programs for this function is within a constant factor of the bounded-error quantum query complexity for the same function:

\begin{theorem}[\cite{spqqc1,spqqc2}]
For any function $f: \D \to \zo$ where $\D \subseteq \zon$, let $Q(f)$ be the bounded-error quantum query complexity of $f$. Then
\beq
Q(f) = \tht{\inf_{\P:{f_{\P}}_{|\D}=f} \ws(\P,\D)},
\eeq
where the infimum is over span programs $\P$ that compute a function agreeing with $f$ on $\D$. Moreover, this infimum is achieved.
\label{thm:wsqq}
\end{theorem}

In particular, a span program with small witness size can be converted into a  quantum algorithm with small query complexity:
\begin{corollary}
For any function $f: \D \to \zo$ where $\D \subseteq \zon$, if $\P$ is a span program computing a function agreeing with $f$ on $\D$, then there exists a bounded-error quantum algorithm that evaluates $f$ with $O(\ws(\P,\D))$ queries.
\label{cor:wsqqc}
\end{corollary}

\section{Span program for tree detection}
\label{sec:sp}

In this section, we build a span program for the subgraph/not-a-minor problem for any tree. This span program has witness size $\bgo{n}$. Then it follows from  Corollary \ref{cor:wsqqc} that this problem has $\bgo{n}$ quantum query complexity.

\begin{theorem}
Let $T=(V_T,E_T)$ be an arbitrary tree. Then there exists a bounded-error quantum algorithm for the subgraph/not-a-minor problem for $T$ with $\bgo{n}$ query complexity.
\label{thm:query}
\end{theorem}

\begin{proof}
Let $G=(V_G,E_G)$ be a graph with $n$ vertices. We need to decide whether $G$ contains $T$ as a subgraph or $G$ does not contain $T$ as a minor, under the promise that one of the cases holds.\\

\noindent\textbf{Color coding.} We use the color coding technique from \cite{colorcoding}. Namely, we map each vertex $u \in V_G$ to a uniformly random node $c(u) \in V_T$, and the vertices of $G$ are colored independently. Then, we discard all the ``badly" colored edges, i.e. we remove any edge $(u,v) \in E_G$ such that $(c(u),c(v))\not\in E_T$. Let $c:V_G \to V_T$ be a random coloring, and let $G_c=(V_G,E_{G_c})$ be the colored graph corresponding to $c$ \footnote{Note that we can determine whether an edge is present in $G_c$ or not by querying the presence of this edge in $G$ and using the information about $c$.}.

We say that $G_c$ contains a correctly colored $T$-subgraph if
if there is an injection $\iota:V_T \to V_G$ such that: (1) $c \circ \iota$ is the identity, i.e. $c(\iota(a))=a$ for any $a \in V_T$; (2) for any $x,y \in V_T$, if $(x,y) \in E_T$, then $(\iota(x),\iota(y))\in E_{G_c}$. 

We will construct a span program that accepts if $G_c$ contains a correctly colored $T$-subgraph, and rejects if $G_c$ does not contain $T$ as a minor. 
Note that if $G$ contains $T$ as a subgraph, then this subgraph is colored correctly with probability at least $|V_T|^{-|V_T|}=\Omega(1)$. So $G_c$ contains $T$ as a subgraph with constant probability. On the other hand, if $G$ does not contain $T$ as a minor, then $G_c$ does not contain $T$ as a minor either. Thus, evaluating our span program for a constant number of independent colorings would suffice to detect $T$ with probability at least $2/3$. \\

\noindent\textbf{Span program.} Our span program $\P$ is defined over the $|V_{T,l}|\lb n+2\rb$-dimensional space spanned by the vectors
\beq
\{ \ket{u} \tp \ket{f}: u \in \{s,t\} \cup V_G, f \in V_{T,l} \},
\eeq
where $\braket{u}{v}=\delta_{u,v}$, for any $u,v \in \{s,t\}\cup V_G$, and $\braket{f}{g}=\delta_{f,g}$, for any $f,g \in V_{T,l}$. 

The target vector of $\P$ is 
\beq
\ket{\tau} \deq \lb \ket{s}-\ket{t} \rb \tp \lb \sl_{f \in V_{T,l}}{\ket{f}} \rb.
\eeq 

The input vectors of $\P$ include the following ones:
\bit
\item For any $u \in c^{-1}(r)$, there is a free input vector 
\beq
\ket{(s,u)} \deq \lb \ket{s}-\ket{u} \rb \tp \lb \sl_{f \in V_{T,l}}{\ket{f}} \rb.
\eeq 
\item For any $f \in V_{T,l}$ and $u \in c^{-1}(f)$, there is a free input vector 
\beq
\ket{(u,t)}\deq (\ket{u}-\ket{t})\tp \ket{f}).
\eeq 
\item For any $x \in V_{T,i}$, $y \in C(x)$, $u \in c^{-1}(x)$
and $v \in c^{-1}(y)$, there is an input vector 
\beq
\ket{(u,v)}\deq \lb \ket{u}-\ket{v} \rb \tp \lb \sl_{f \in V^{y}_{T,l}}{\ket{f}} \rb,
\eeq
and this input vector is available if and only if $(u,v) \in E_{G_c}$.\\
\eit

Here is a more compact way to describe these input vectors. We add two special vertices $s$ and $t$ to $G_c$, and color $s$, $t$ as themselves, i.e. $c(s)=s$ and $c(t)=t$. Furthermore, we connect $s$ to the vertices in $c^{-1}(r)$, and connect $t$ to the vertices in $c^{-1}(f)$ for any $f \in V_{T,l}$. Let $G'_c$ be this modified graph. For any $x \in V_T \cup \{s,t\}$, we call $c^{-1}(x)$ a \emph{block}. Then $G'_c$ contains only edges between \emph{adjacent} blocks. Namely, $c^{-1}(x)$ and $c^{-1}(y)$ are adjacent if and only if $x$ and $y$ are adjacent, i.e. $x \sim y$. Then, for any $u,v$ in adjacent blocks, we have an input vector
\beq
\ket{(u,v)}=(\ket{u}-\ket{v}) \tp \lb \sum_{f \in S(c(u),c(v))}\ket{f} \rb,
\label{eq:inputvectors}
\eeq
and this input vector is available if and only if the edge $(u,v)$ is present in $G'_c$. \\

\noindent\textbf{An example.}
Let $T$ be the complete $2$-level binary tree, and let $G$ be a $12$-vertex graph shown in Fig.\ref{fig:goodG}. Let
$c:V_G \to V_T$ be a coloring defined as
$c(u_1)=c(u_2)=r$,
$c(u_3)=c(u_4)=d_1$, $c(u_5)=c(u_6)=d_2$, $c(u_7)=f_1$, $c(u_8)=c(u_9)=f_2$,
$c(u_{10})=c(u_{11})=f_3$, $c(u_{12})=f_4$. Then, after removing the badly colored edges (such as $(u_2,u_8)$), the colored graph $G_c$ is shown in Fig.\ref{fig:goodGc}.

\begin{figure}[H]
\begin{center}
\subfigure[] 
{\includegraphics[scale=0.5]{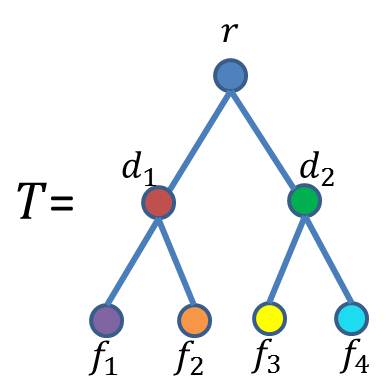}}
\label{fig:goodT}
\subfigure[] 
{\includegraphics[scale=0.5]{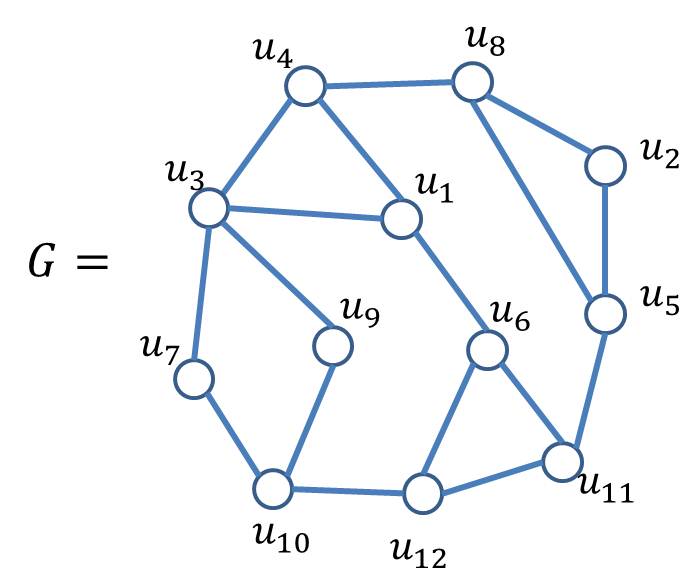}
\label{fig:goodG}}
\subfigure[] 
{\includegraphics[scale=0.5]{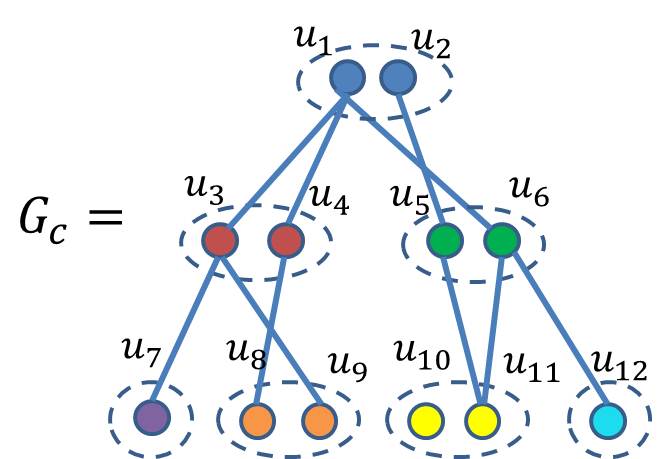}
\label{fig:goodGc}
}
\end{center}
\caption{An example of coloring coding.} 
\label{fig:color}
\end{figure}

Then the span program $\P$ is defined as follows:
\bit
\item Target vector: $(\ket{s}-\ket{t}) \tp (\sum_{j \in [4]}\ket{f_j})$;
\item Free input vectors: 
\bit
\item $(\ket{s}-\ket{u_i}) \tp (\sum_{j \in [4]}\ket{f_j})$, for $i \in [2]$;
\item $(\ket{u_j}-\ket{t}) \tp \ket{f_{\chi(j)}}$,
for $j \in \{7,8,\dots,12\}$, where $\chi(7)=1$, $\chi(8)=\chi(9)=2$, $\chi(10)=\chi(11)=3$, $\chi(12)=4$;
\eit
\item Other input vectors:
\bit
\item $(\ket{u_i}-\ket{u_j}) \tp (\ket{f_1}+\ket{f_2})$ for $i \in \{1,2\}$ and $j \in \{3,4\}$;
\item $(\ket{u_i}-\ket{u_j}) \tp (\ket{f_3}+\ket{f_4})$ for $i \in \{1,2\}$ and $j \in \{5,6\}$;
\item $(\ket{u_i}-\ket{u_7}) \tp \ket{f_1}$ for $i \in \{3,4\}$;
\item $(\ket{u_i}-\ket{u_j}) \tp \ket{f_2}$ for $i \in \{3,4\}$ and $j \in \{8,9\}$;
\item $(\ket{u_i}-\ket{u_j}) \tp \ket{f_3}$ for $i \in \{5,6\}$ and $j \in \{10,11\}$;
\item $(\ket{u_i}-\ket{u_{12}}) \tp \ket{f_4}$ for $i \in \{5,6\}$.
\eit
\eit

\noindent\textbf{Witness size.} Next, we will show that our span program $\P$ indeed solves the subgraph/not-a-minor problem for $T$, and along the way we also obtain upper bounds on the positive and negative witness sizes of $\P$. We will consider the positive and negative cases separately.\\

\noindent\textbf{Positive case.} Let us first consider the positive case, i.e. $G_c$ contains $T$ as a subgraph.

\begin{lemma}
Suppose $G_c$ contains $T$ as a subgraph. Then $\P$ accepts $G_c$. Moreover, the witness size of $\P$ on $G_c$ is $\bgo{1}$.
\label{lem:posws}
\end{lemma}
\begin{proof}
Suppose $G_c$ contains $T$ as a subgraph. Then there exists an injection $\iota: V_T \to V_G$ such that, for any $a \in V_{T,i}$ and $b \in C(a)$, $(\iota(a),\iota(b)) \in E_{G_c}$ and hence the input vector $\ket{(\iota(a),\iota(b))}$ is available. Thus, the target vector $\ket{\tau}$ can be written as
\beq
\ket{\tau}=
\ket{(s,\iota(r))}
+\sl_{a \in V_{T,i}}{\sl_{b \in C(a)}{\ket{(\iota(a),\iota(b))}}}
+\sl_{f \in V_{T,l}}{\ket{(\iota(f),t)}}.
\eeq
So $\P$ accepts $G_c$. Furthermore, the witness size of $\P$ on $G_c$ is $|E_T|=\bgo{1}$, since $T$ has constant size.
\end{proof}

For example, consider the $T$ and $G_c$ in Fig.\ref{fig:color}. $G_c$ contains $T$ as a subgraph, and $\P$ accepts $G_c$. The solution is
\beba
\ket{\tau}&=&\ket{(s,u_1)}+\ket{(u_1,u_3)}+\ket{(u_1,u_6)}
+\ket{(u_3,u_7)}+\ket{(u_3,u_9)}
+\ket{(u_6,u_{11})}+\ket{(u_6,u_{12})}\\
&&+\ket{(u_7,t)}+\ket{(u_9,t)}
+\ket{(u_{11},t)}+\ket{(u_{12},t)}.
\eeea
This solution can be graphically represented by four ``parallel" flows from $s$ to $t$ shown in Fig.\ref{fig:good}. 
\begin{figure}[H]
\center
\includegraphics[scale=0.45]{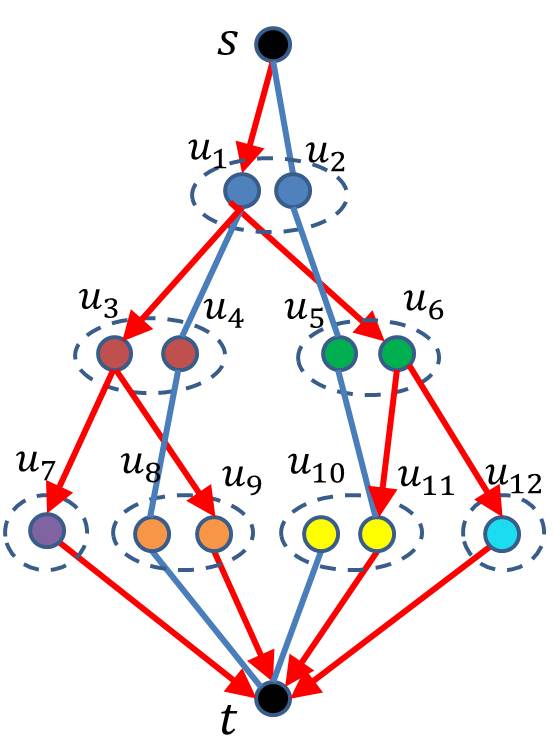}
\caption{A graphical representation of the solution to the span program $\P$ in the positive case. For each edge used by the flows, the coefficient for the corresponding input vector is $+1$. The coefficient for any other input vectors is $0$.}
\label{fig:good}
\end{figure}

\noindent\textbf{Negative case.} Now let us consider the negative case, i.e. $G_c$ does not contain $T$ as a \emph{minor}.

But first, let us explain why we do not just consider the case $G_c$ does not contain $T$ as a \emph{subgraph}. The problem is that the span program $\P$ fails to solve this problem in general. Namely,
\begin{claim}
There exists some $G_c$ which does not contain $T$ as a subgraph but $\P$  accepts it.
\end{claim}
\begin{proof}
For example, consider the $T$ and $G_c$ in Fig.\ref{fig:bad}. Here $T$ is the complete $2$-level binary tree, and $G_c$ is a $12$-vertex graph shown in Fig.\ref{fig:badGc}. $G_c$ does not contain $T$ as a subgraph, but $\P$ accepts $G_c$. The solution is 
\beba
\ket{\tau}&=&
\ket{(s,u_1)}-\ket{(s,u_2)}+\ket{(s,u_3)}
+\ket{(u_1,u_5)}+\ket{(u_1,u_6)}-\ket{(u_2,u_6)}
-\ket{(u_2,u_6)}
+\ket{(u_3,u_7)}\\
&&+\ket{(u_3,u_8)}
+\ket{(u_5,u_9)}
+\ket{(u_5,u_{10})}
+\ket{(u_8,u_{11})}
+\ket{(u_8,u_{12})}
+\ket{(u_9,t)}
+\ket{(u_{10},t)}\\
&&+\ket{(u_{11},t)}
+\ket{(u_{12},t)}.
\eeea
This solution can be graphically represented as four flows from $s$ to $t$ shown in Fig.\ref{fig:badsol}. Notice that these flows move back and forth between adjacent ``layers", where each layer is colored by the nodes at the same depth of $T$. Also, note that if we contract the edges $(u_1,u_6)$, $(u_2,u_6)$, $(u_2,u_7)$ and $(u_3,u_7)$ of $G_c$, we would get a new graph isomorphic to $T$. In other words, $G_c$ contains $T$ as a \emph{minor}.
\begin{figure}[H]
\center
\subfigure[]
{\includegraphics[scale=0.5]{T2.jpg}}
\label{fig:badT}
\subfigure[]
{\includegraphics[scale=0.55]{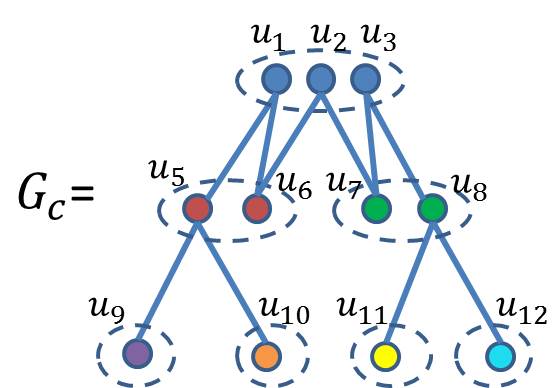}}
\label{fig:badGc}
\subfigure[]
{\includegraphics[scale=0.55]{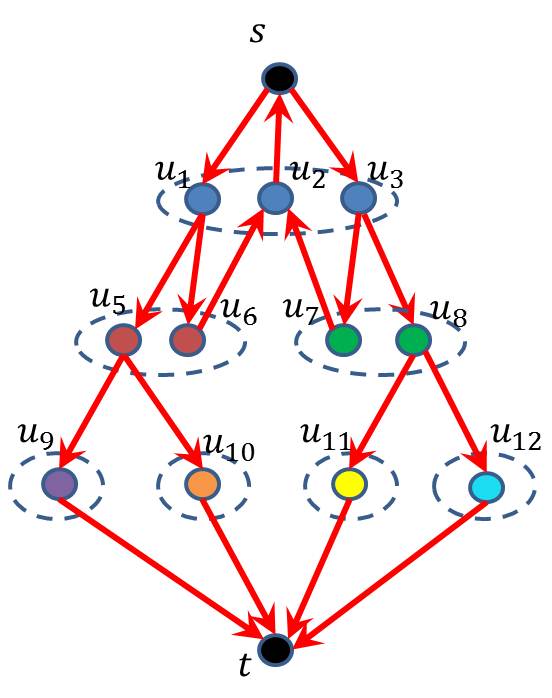}}
\label{fig:badsol}
\caption{An example showing that the span program $\P$ does not solve the tree containment problem in general. Here $G_c$ does not contain $T$ as a subgraph, but $\P$ accepts it. The solution is graphically illustrated by Fig.\ref{fig:badsol}. For each edge used by the flows, the coefficient for the corresponding input vector is $+1$ or $-1$, depending on whether the flow moves towards $t$ or $s$, respectively. The coefficient for any other input vector is $0$.}
\end{figure}
\label{fig:bad}
\end{proof} 

We can directly prove that if $\P$ accepts $G_c$, then $G_c$ must contain $T$ as a minor. But here we prefer to use a different approach. We will assume that $G_c$ does not contain $T$ as a minor. Then, we will show that $\P$ must reject $G_c$ by explicitly giving a negative witness for $G_c$. The advantage of this approach is that we not only get to know that $\P$ accepts only certain graphs that contain $T$ as a minor, but also obtain an upper bound on the negative witness size of $\P$.

Now let us return to the negative case, i.e. $G_c$ does not contain $T$ as a minor. For convenience, we introduce the following notation. For any $(a,b) \in E_T$ and $u \in c^{-1}(a)$, let $N_{a \to b}(u)=\{v \in c^{-1}(b): (u,v) \in E_{G_c}\}$. For any $U \subseteq c^{-1}(a)$, let $N_{a \to b}(U)=\cup_{u \in U} N_{a \to b}(u)$. For any $a \in V_T$, let $Y^a_c = c^{-1}(V^a_T) = \cup_{b \in V^a_T} c^{-1}(b)$ and let $G^a_c = G_c|_{Y^a_c}$.

We claim that $G_c$ must satisfy the following property:

\begin{claim}
For any graph $G=(V_G,E_G)$ and coloring $c:V_G \to V_T$, if $G_c=(V_G,E_{G_c})$ does not contain $T$ as a minor, then there exist $\{V_a \subseteq c^{-1}(a): a \in V_{T,i}\}$ and $\{V_{a,b} \subseteq c^{-1}(a): a \in V_{T,i}, b \in C(a)\}$ such that
\ben
\item 
$\forall a \in V_{T,i}$, $V_a$ is the disjoint union of $V_{a,b}$'s for $b \in C(a)$; 
\item
$V_r=c^{-1}(r)$;
\item 
$\forall a \in V_{T,i}$, $\forall b \in C(a) \cap V_{T,i}$, 
$N_{a \to b}(V_{a,b}) \subseteq V_b$
and
$N_{b \to a}(V_b) \subseteq V_{a,b}$;
\item
$\forall a \in V_{T,i}$, $\forall b \in C(a) \cap V_{T,l}$, 
$N_{a \to b}(V_{a,b})=\varnothing$.
\een
\label{clm:structure}
\end{claim}
\begin{proof}
See Appendix \ref{apd:pfclm}.
\end{proof}

Intuitively, the $V_a$'s contain the ``bad" vertices that are responsible for the fact that $G_c$ does not contain $T$ as a minor. For any vertex $u \in V_a$, there is no subgraph of $G^a_c$ that can be contracted into a tree rooted at $u$ and isomorphic to $T^a$. In particular, since $G_c$ does not contain $T$ as a minor, for any vertex $u \in c^{-1}(r)$, there is no subgraph of $G_c$ that can be contracted into a tree rooted at $u$ and isomorphic to $T$, and hence $V_r=c^{-1}(r)$. Furthermore, if $u \in V_{a,b} \subseteq V_a$ for some $b \in C(a)$, then $u$ is ``bad" because its ``children" in $c^{-1}(b)$ are ``bad". Namely, if we attempt to use $u$ as the root of $T^a$, then we will fail because we will not be able to get a {\emph {complete}} $T^b$ (which is a subtree of $T^a$). In the degenerate case, $b$ is a leaf, and $u \in V_{a,b}$ if and only if it has no neighbor in $c^{-1}(b)$. Condition $3$ tells us that these ``bad" vertices form a ``connected component" in some sense, and it is crucial for bounding the witness size. Fig.\ref{fig:clm1} demonstrates an example of such $V_a$'s and $V_{a,b}$'s. 

\begin{figure}[H]
\center
\subfigure[]
{\includegraphics[scale=0.35]{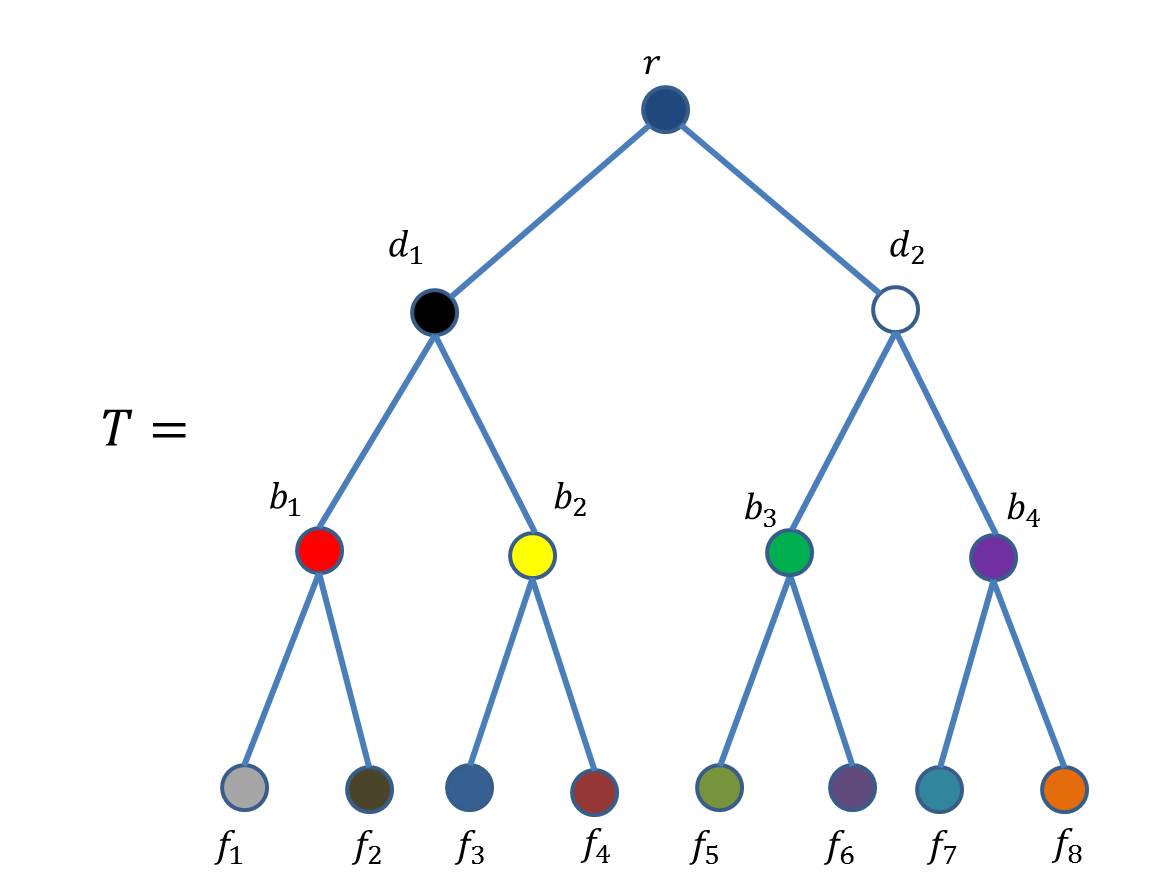}}
\subfigure[]
{\includegraphics[scale=0.45]{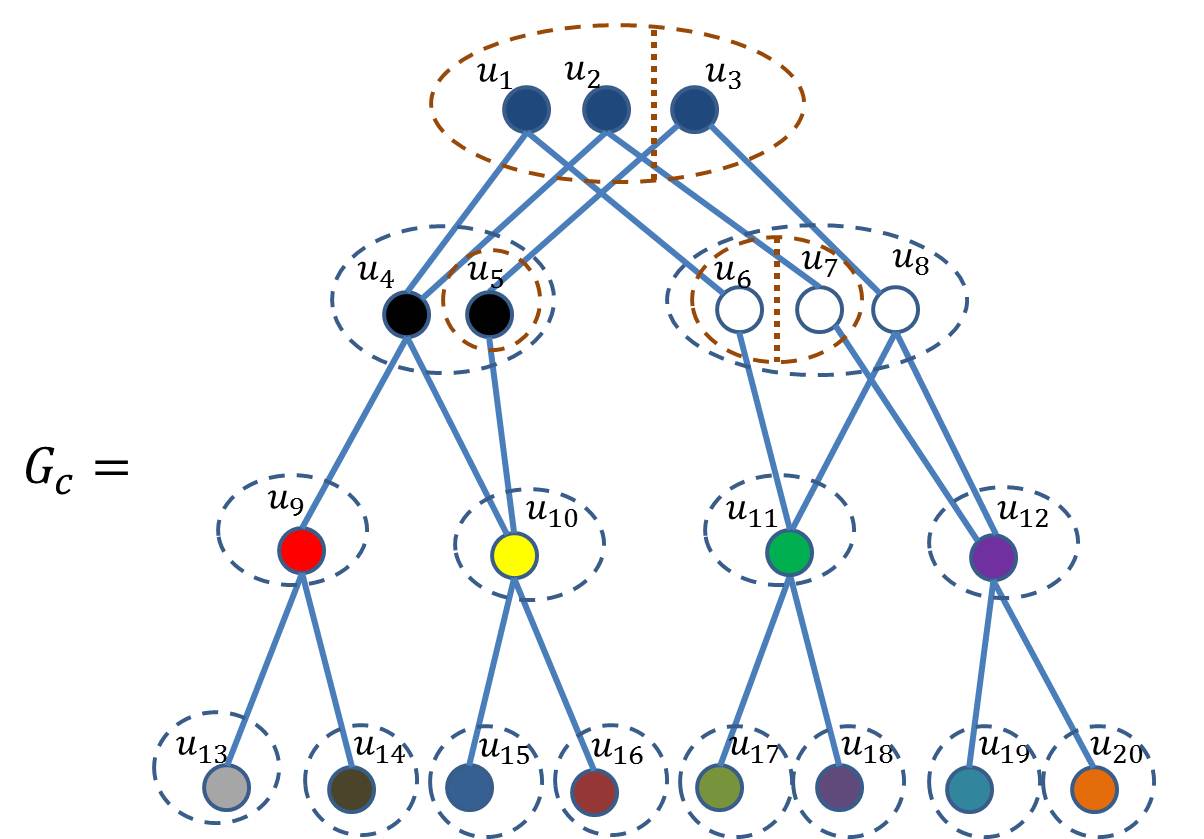}}
\caption{In this example, $T$ is the full $3$-level binary tree, and $G_c$ does not contain $T$ as a minor. We have $V_r=\{u_1,u_2,u_3\}$, $V_{r,d_1}=\{u_3\}$, 
$V_{r,d_2}=\{u_1,u_2\}$, $V_{d_1}=V_{d_1,b_1}=\{u_5\}$,
$V_{d_2}=\{u_6,u_7\}$, $V_{d_2,b_3}=\{u_7\}$, $V_{d_2,b_4}=\{u_6\}$,
and any other $V_a$ or $V_{a,b}$ is $\varnothing$. Note that, 
$N_{r \to d_2}(V_{r,d_2})=\{u_6,u_7\}=V_{d_2}$,
$N_{d_2 \to b_3}(V_{d_2,b_3})=\varnothing=V_{b_3}$
and
$N_{d_2 \to b_4}(V_{d_2,b_4})=\varnothing=V_{b_4}$, etc.
}
\label{fig:clm1}
\end{figure}

\begin{lemma}
Suppose $G_c$ does not contain $T$ as a minor. Then $\P$ rejects it. Moreover, the witness size of $\P$ on $G_c$ is $\bgo{n^2}$.
\label{lem:negws}
\end{lemma}
\begin{proof}
Using Claim \ref{clm:structure}, we build a negative witness for $G_c$ as follows. Let 
\beq
\ket{w}=\ket{w_s}+\sum_{a \in V_{T,i}} \ket{w_a},
\eeq 
where 
\beba
\ket{w_s}=\ket{s} \tp \lb \dfrac{1}{|V_{T,l}|} \sl_{f \in V_{T,l}} {\ket{f}} \rb, \\
\ket{w_a}=\sl_{b \in C(a)}{\ket{w_{a,b}}}
\eeea
where
\beq
\ket{w_{a,b}}= \lb \sl_{u \in V_{a,b}}{\ket{u}} \rb \tp
\lb \dfrac{1}{|V^b_{T,l}|} \sl_{f \in V^b_{T,l}}{\ket{f}} \rb.
\eeq

Now we prove that $\ket{w}$ is a valid negative witness. First,
\beq
\braket{w}{\tau}=\braket{w_s}{\tau}+
\sum_{a \in V_{T,i}}\sum_{b \in C(a)} \braket{w_{a,b}}{\tau}
=1+0=1.
\eeq 
Next, we show that $\ket{w}$ is orthogonal to all available input vectors on $G_c$. We deal with four kinds of available input vectors separately:

\ben
\item For any free input vector $\ket{(s,u)}=(\ket{s}-\ket{u}) \tp (\sum_{f \in V_{T,l}} \ket{f})$ where $u \in c^{-1}(r)$: By conditions 1 and 2 of Claim \ref{clm:structure}, $c^{-1}(r)=V_r=\cup_{b \in C(r)} V_{r,b}$, so we can find $b \in C(r)$ such that $u \in V_{r,b}$. Then we have $\braket{w}{(s,u)}=0$, since
$\braket{w_s}{(s,u)}=1$, $\braket{w_{r,b}}{(s,u)}=-1$, and $\braket{w_{a',b'}}{(s,u)}=0$ for any other $(a',b')$.

\item For any free input vector $\ket{(v,t)}=(\ket{v}-\ket{t}) \tp \ket{f}$ where $v \in c^{-1}(f)$ for some $f \in V_{T,l}$: We have $\braket{w}{(v,t)}=0$, since $\braket{w_s}{(v,t)}=\braket{w_{a,b}}{(v,t)}=0$ for any $(a,b)$.

\item For any available input vector $\ket{(u,v)}=(\ket{u}-\ket{v}) \tp (\sum_{f \in V^b_{T,l}} \ket{f})$ where $u \in c^{-1}(a)$, $v \in c^{-1}(b)$ for some $a \in V_{T,i}$ and $b \in C(a) \cap V_{T,i}$: The availability of this input vector implies $(u,v) \in E_{G_c}$. Then, by condition 3 of Claim \ref{clm:structure}, there are two possible cases:
\bit
\item $u \in V_{a,b}$, $v \in V_{b}$: In this case, $\braket{w_{a,b}}{(u,v)}=1$. Also, by condition 1 of Claim \ref{clm:structure}, there exists $d \in C(b)$ such that $v \in V_{b,d}$. Then $\braket{w_{b,d}}{(u,v)}=-1$, and $\braket{w_{s}}{(u,v)}=\braket{w_{a',b'}}{(u,v)}=0$ for any other $(a',b')$. Thus,
$\braket{w}{(u,v)}=0$.
\item $u \not\in V_{a,b}$, $v \not\in V_b$: In this case, we simply have
$\braket{w_{s}}{(u,v)}=\braket{w_{a',b'}}{(u,v)}=0$ for any $(a',b')$. Hence, $\braket{w}{(u,v)}=0$.
\eit
\item For any available input vector $\ket{(u,v)}=(\ket{u}-\ket{v}) \tp \ket{f}$ where $u \in c^{-1}(a)$, $v \in c^{-1}(f)$ for some $a \in V_{T,i}$ and $f \in C(a) \cap V_{T,l}$: The availability of this input vector implies $(u,v) \in E_{G_c}$. Then, by condition 4 of Claim 1, we must have
$u \not\in V_{a,f}$. It follows that $\braket{w_s}{(u,v)}=\braket{w_{a',b'}}{(u,v)}=0$ for any $(a',b')$, and hence $\braket{w}{(u,v)}=0$.
\een

Now, note that $\ket{w}$ can be written as
\beq
\ket{w}=\sum_{u \in \{s\}\cup V_G} \sum_{f \in V_{T,l}}
\mu_{u,f} \ket{u}\tp \ket{f}
\eeq 
where
$|\mu_{u,f}|\le 1$ for any $(u,f)$. Meanwhile, every input vector of $\P$ is the sum of a constant number of some $\pm \ket{u} \tp \ket{f}$'s. Thus, the inner product between $\ket{w}$ and any input vector has norm $O(1)$. Since there are $O(n^2)$ input vectors in $\P$, the negative witness size of $\P$ on $G_c$ is at most $O(n^2)$.
\end{proof}

Combining Lemma \ref{lem:posws} and Lemma \ref{lem:negws} together, we know that $\P$ solves the subgraph/not-a-minor problem for $T$, and it has witness size $O(n)$. Then, by Corollary \ref{cor:wsqqc}, this problem can be solved by a bounded-error quantum algorithm with $O(n)$ query complexity.
\end{proof}

\section{Time-efficient implementation} 

Theorem \ref{thm:query} implies the existence of a \emph{query-efficient} quantum algorithm for the subgroup/not-a-minor problem for any tree. However, this algorithm is not necessarily \emph{time-efficient}. Indeed, any span program can be evaluated by running phase estimation \cite{pe,pe2} on some quantum walk operator associated with this span program \cite{spqqc2,spstconnect}. But this approach does not always yield time-efficient quantum algorithms, because this quantum walk operator can be difficult to implement in general. Using some ideas from \cite{spstconnect}, we nevertheless manage to implement this operator for our span program $\P$ quickly and hence give a time-efficient implementation of the algorithm from Theorem \ref{thm:query}.

\begin{theorem}
The algorithm from Theorem \ref{thm:query} can be implemented in $\tbgo{n}$ quantum time.
\label{thm:time}
\end{theorem}

\begin{proof}
We use the approach of \cite{spqqc2,spstconnect} for evaluating span programs. This approach relies on an ``effective" spectral gap of a quantum walk operator associated with the span program. Specifically, suppose $\Q$ is an arbitrary span program with input vectors $\ket{v_1},\dots,\ket{v_k} \in \R^d$ and target vector $\ket{\tau} \in  \R^d$. Let $\D \subseteq \zon$, and let $W_1=\ws_1(\Q,\D)$, $W_0=\ws_0(\Q,\D)$, $W=\ws(\Q,\D)$ be the positive, negative and overall witness size of $\Q$ over domain $\D$, respectively. We assume that $W_1, W_0, W,k=\poly{n}$. Pick a constant $C>\mmax{10,1/W}$. Let $\alpha=C\sqrt{W_1}$ and  $\ket{\tilde{\tau}}=\ket{\tau}/\alpha$. Then define
\beq
V \deq \ketbra{\tilde{\tau}}{0}+\sum_{j \in [k]} \ketbra{v_j}{j}.
\eeq
Let $R_{\Lambda}=2\Lambda-I$, where $\Lambda$ is the projection onto $\ker{V}$. Moreover, for any $x \in \D$, let $R_{x}=2\Pi(x)-I$, where $\Pi(x) \deq \sum_{j \in \{0\}\cup I(x)} \ketbra{j}{j}$, where $I(x)$ is the index set of available input vectors on input $x$. The algorithm for evaluating $\Q$ work in the Hilbert space $\H \deq \span{\ket{0},\ket{1},\dots,\ket{k}}$. On input $x$, it starts in $\ket{0}$ and runs phase estimation on $U \deq R_{\Lambda}R_{x}$ with precision $1/(10CW)$ and error rate $1/10$, and it accepts if and only if the measured phase is $0$. This algorithm uses $O(W)$ controlled applications of $U$, and it correctly evaluates $\Q$ on $x$ with probability at least $2/3$. The key component of this algorithm is the implementation of $R_{\Lambda}$ and $R_{x}$. Usually $R_{x}$ can be implemented with few input queries and $\polylog{n}$ local gates. But $R_{\Lambda}$ can be much harder to implement. Let $T_0$ and $T_1$ be the time required to implement $R_x$ and $R_{\Lambda}$ respectively. Then the time complexity of this algorithm is $\tbgo{W (T_0+T_1)}$.

Now we apply this general approach to our span program $\P$ for tree detection. To efficiently implement the reflection $R_{\Lambda}$, we invoke the following lemma:

\begin{lemma}[Spectral lemma \cite{qw}]
Let $A$ and $B$ be complex matrices such that they have the same number of rows and each of them has orthonormal columns. Let $D(A,B) = A^{\dagger}B$, and let $U(A,B) = \refl{B} \cdot \refl{A}$. Then all the singular values of $D(A,B)$ are at most $1$. Let $\cos{\theta_1}$, $\cos{\theta_2}$, $\dots$, $\cos{\theta_l}$ be the singular values of $D(A,B)$ that lie in the open interval $(0,1)$ counted with multiplicity. Then the following is a complete list of the eigenvalues of $U(A,B)$:
\ben
\item The $+1$ eigenspace of $U(A,B)$ is $\lb \Col{A} \cap \Col{B} \rb \oplus \lb \Col{A}^{\bot} \cap \Col{B}^{\bot}\rb$;
\item The $-1$ eigenspace of $U(A,B)$ is $\lb \Col{A} \cap \Col{B}^{\bot} \rb \oplus \lb \Col{A}^{\bot} \cap \Col{B} \rb$;
\item The other eigenvalues of $U(A,B)$ are $e^{2i\theta_1}, e^{-2i\theta_1}, e^{2i\theta_2}, e^{-2i\theta_2},\dots,
e^{2i\theta_l}, e^{-2i\theta_l}$ counted with multiplicity.
\een
\label{lem:spectral}
\end{lemma}

Following the idea of \cite{spstconnect}, we will find two matrices $A$ and $B$ such that: (1) they have the same number of rows; (2) each of them has orthonormal columns; (3) $V' \deq A^{\dagger}B= \frac{1}{\sqrt{4n}}V$. Then, Lemma \ref{lem:spectral} implies that the $-1$ eigenspace of $U(A,B)=\refl{B} \cdot \refl{A}$ is $\lb \Col{A} \cap \Col{B}^{\bot} \rb \oplus \lb \Col{A}^{\bot} \cap \Col{B} \rb$. Note that $\Col{A}^{\bot} \cap \Col{B}=B(\ker{V'})=B(\ker{V})$, and $\Col{A} \cap \Col{B}^{\bot}$ is orthogonal to $\Col{B}$. Thus, to ``effectively" implement $R_{\Lambda}$, we embed $\H$ into $B(\H)$ and treat the $-1$ eigenspace of $U(A,B)$ as $\ker{V}$. That is, we simulate the behavior of $R_{\Lambda}$ on any $\ket{\phi} \in \H$ by the behavior of $R_{-1}$ on $B(\ket{\phi}) \in B(\H)$, where $R_{-1}$ is the reflection about the $-1$ eigenspace of $U(A,B)$. This simulation is valid because $B$ is an isometry. Now, $R_{-1}$ can be (approximately) implemented by running phase estimation on $U(A,B)$ and multiplying the phase by $-1$ if the measured eigenvalue is very close to $-1$. The precision of this phase estimation depends on the eigenvalue gap around $-1$ of $U(A,B)$. Let $T_A$ and $T_B$ be the time required to implement $R_A$ and $R_B$ respectively, and let $\delta_{A,B}$ is the eigenvalue gap around $-1$ of $U(A,B)$. Then $R_{\Lambda}$ can be implemented in time $\tbgo{(T_A+T_B)/\delta_{A,B}}$. 

Next, we will describe how to factorize $V$ into $A$ and $B$. But before doing that, we need to modify our span program $\P$ to make it possess a more uniform structure. Specifically, we modify $G'_c$ and $\P$ as follows: 
\ben

\item We add dummy vertices to each block of $G'_c$, so that each block contains exactly $n$ vertices. Then we fill in the graph with never-available edges between adjacent blocks, so that there is a complete bipartite graph between any two adjacent blocks. 

\item We normalize each input vector to unit length. Specifically, for an input vector $\ket{(u,v)}$, we scale it by a factor of ${1}/{\sqrt{2|S(c(u),c(v))|}}$.

\item We scale the target vector $\ket{\tau}$ by a factor of $1/\sqrt{|V_{T,l}|}$ so that it has length $\sqrt{2}$.

\item We pick a constant $C>\mmax{10,1/W,1/\sqrt{W_1}}$. Let
$\alpha=C\sqrt{W_1}>1$ and
\beq
\ket{\tilde{\tau}}=\frac{1}{\alpha}\lb\ket{s}-\ket{t}\rb\tp \lb\frac{1}{\sqrt{|V_{T,l}|}} \sum_{f \in V_{T,l}}\ket{f}\rb.
\eeq
We add a never-available input vector 
\beq
\ket{\gamma} \deq \sqrt{1-\frac{1}{\alpha^2}}\lb\ket{t}-\ket{s}\rb\tp \lb\frac{1}{\sqrt{|V_{T,l}|}} \sum_{f \in V_{T,l}}\ket{f}\rb.
\eeq
\een

It is easy to check that this modified span program still computes the same function, and its positive and negative witness size remain $\bgo{1}$ and $\bgo{n^2}$ respectively.

Now, since each block of $G'_c$ contains $n$ vertices, we represent the vertices and edges of $G'_c$ as follows. We denote the $k$-th vertex in the block $c^{-1}(x)$ as $(x,k)$. In particular, we denote the original $s$ as $(s,1)$ and denote the original $t$ as $(t,1)$. Then, for the edge between the vertices $(x_1,k_1)$ and $(x_2,k_2)$, we denote it as $(x_1,k_1,x_2,k_2)$. But this leads to a problem. Namely, this edge also has another representation -- $(x_2,k_2,x_1,k_1)$. This could make the implementation of $\refl{A}$ and $\refl{B}$ a little harder. We solve this problem by making two copies of each input vector (except $\ket{\gamma}$), so that one copy corresponds to the representation $(x_1,k_1,x_2,k_2)$ and the other corresponds to the respresentation $(x_2,k_2,x_1,k_1)$. Furthermore, we make $\ket{\tilde{\tau}}$ and $\ket{\gamma}$ correspond to $(s,1,t,1)$ and $(t,1,s,1)$ respectively. 

Now define 
\beq
I \deq \{(x,k,f):x \in V_T \cup \{s,t\}, k \in [n],f \in V^x_{T,l}\}
\eeq
and 
\beq
J \deq \{(x_1,k_1,x_2,k_2): x_1,x_2 \in V_T \cup \{s,t\}, x_1 \sim x_2,~~ k_1,k_2 \in [n]\}.
\eeq
Then for our span program $\P$, we have
\beq
V=\sl_{j\in J}\ketbra{v_j}{j}
\eeq
where 
\beq
\ket{v_j}=\bc
\dfrac{1}{\sqrt{2}}\lb \ket{x,k}-\ket{y,k'} \rb \tp \lb \dfrac{1}{\sqrt{|S_{x,y}|}}
\sl_{f \in S_{x,y}} \ket{f} \rb,&~\mathrm{if}~j=(x,k,y,k')\not\in \{(s,1,t,1), (t,1,s,1)\},\\
\dfrac{1}{\alpha}\lb\ket{s,1}-\ket{t,1}\rb\tp \lb\dfrac{1}{\sqrt{|V_{T,l}|}} \sl_{f \in V_{T,l}}\ket{f}\rb,&~\mathrm{if}~j=(s,1,t,1),\\
\sqrt{1-\dfrac{1}{\alpha^2}}\lb\ket{t,1}-\ket{s,1}\rb\tp \lb\dfrac{1}{\sqrt{|V_{T,l}|}} \sl_{f \in V_{T,l}}\ket{f}\rb,&~\mathrm{if}~j=(t,1,s,1).
\ec
\eeq\\

Now we are ready to give the factorization of $V$ into $A$ and $B$. We will find unit vectors $\{\ket{a_i}: i \in I\}$ and $\{\ket{b_j}: j\in J\}$ such that
\beq
\braket{a_i}{j}\braket{i}{b_j}=\dfrac{1}{\sqrt{4n}}\bra{i}V\ket{j},
\label{eq:aibj}
\eeq
Then we define
\beba
A=\sl_{i \in I}(\ket{i}\tp \ket{a_i})\bra{i},\\
B=\sl_{j \in J} (\ket{b_j}\tp \ket{j})\bra{j}.
\eeea
It follows immediately that
\beq
V'\deq A^{\dagger}B=\dfrac{V}{\sqrt{4n}},
\eeq
as desired. The $\ket{a_i}$'s are defined as follows:
\bit
\item For $i=(x,k,f)$, where $x\not\in \{s,t\}$ or $k \neq 1$, let 
\beba
\ket{a_i} &=& \dfrac{1}{\sqrt{4n}}\sl_{y \in M(x,f)} \sl_{k'\in [n]}\lb \ket{x,k,y,k'}+\ket{y,k',x,k}\rb. 
\label{eq:ai1}
\eeea
Note that since $|M(x,f)|=2$, $\ket{a_i}$ is indeed a unit vector. 
\item For $i=(s,1,f)$, let 
\beba
\ket{a_i} &= &\dfrac{1}{\sqrt{4n}}\lb\sl_{k \in [n]}(\ket{s,1,r,k}+\ket{r,k,s,1}) + \sl_{k=2}^n (\ket{s,1,t,k}+ \ket{t,k,s,1})\rb \\
&&+\dfrac{1}{\sqrt{2n}}\lb
\dfrac{1}{\alpha}\ket{s,1,t,1}+\sqrt{1-\dfrac{1}{\alpha^2}}\ket{t,1,s,1}\rb.
\label{eq:ai2}
\eeea
\item 
For $i=(t,1,f)$, let
\beba
\ket{a_i} &= &\dfrac{1}{\sqrt{4n}}\lb \sl_{k \in [n]}(\ket{t,1,f,k}+\ket{f,k,t,1})
+\sl_{k=2}^n (\ket{t,1,s,k}+\ket{s,k,t,1})\rb \\
&&+\dfrac{1}{\sqrt{2n}}\lb
\dfrac{1}{\alpha}\ket{s,1,t,1}+\sqrt{1-\dfrac{1}{\alpha^2}}\ket{t,1,s,1}\rb.
\label{eq:ai3}
\eeea
\eit
The $\ket{b_j}$'s are defined as follows:
\bit
\item For $j=(x_1,k_1,x_2,k_2)$, let 
\beq
\ket{b_j} = \dfrac{1}{\sqrt{2|S_{x_1,x_2}|}} \sl_{f\in S_{x_1,x_2}}\lb \ket{x_1,k_1,f}-\ket{x_2,k_2,f}\rb .
\label{eq:bj}
\eeq
\eit
It is easy to check that these $\ket{a_i}$'s and $\ket{b_j}$'s are indeed unit vectors and they satisfy Eq.(\ref{eq:aibj}), as promised.

Now we describe the implementation of $R_x$ and $R_{\Lambda}$. We embed the space $\H=\span{\ket{j}:~j \in J}$ into $B(\H)=\span{\ket{b_j}\ket{j}:~j \in J}$. To implement $R_x$, we detect $j$ in the second register and multiply the phase by $-1$ if $\ket{v_j}$ is an unavailable input vector on $x$ (this can be tested by querying the edge labelled by $j$). To implement $R_{\Lambda}$, we run phase estimation on the quantum walk operator $U(A,B)=\refl{B} \cdot \refl{A}$ with precision $\delta_{A,B}/3$ and multiply the phase by $-1$ if the measured eigenvalue is $(\delta_{A,B}/3)$-close to $-1$. It remains to analyze the spectral gap $\delta_{A,B}$ around $-1$ of $U(A,B)$ and give explicit implementation of $\refl{A}$ and $\refl{B}$.

\begin{lemma}
The eigenvalue gap $\delta_{A,B}$ around $-1$ of $U(A,B)=\refl{B} \cdot \refl{A}$ is $\omg{1}$.
\label{lem:gap}
\end{lemma}
\begin{proof}
By Lemma \ref{lem:spectral}, it is sufficient to show that the singular value gap around $0$ of $V'=A^{\dagger}B$ is $\omg{1}$. To prove this, it is sufficient to show that the smallest non-zero eigenvalue of 
\beq
\Delta\deq {V'}{V'}^{\dagger}=\dfrac{1}{4n}\sl_{j\in J} \ketbra{v_j}{v_j}
\eeq 
is $\Omega(1)$. Let $F=\{(x,y):~x,y \in V_T \cup \{s,t\}, ~x  \sim y \}$, and let $F(x,y)=\{(x,k,y,k'):~k,k \in [n]\}$ for any $(x,y) \in F$. Then we can write $\Delta$ as 
\beq
\Delta=\sl_{(x,y) \in F} \Delta_{x,y}
\eeq
where
\beq
\Delta_{x,y} \deq \frac{1}{4n} \sum_{j \in F(x,y)} \ketbra{v_j}{v_j}
\eeq
Now, for any $j=(x,k,y,k')$, where $(x,y) \not\in \{(s,t),(t,s)\}$, we have 
\beba
\ketbra{v_j}{v_j}
&=&\dfrac{1}{2|S_{x,y}|} \sl_{f,f' \in S_{x,y}} \lb \ketbra{x,k,f}{x,k,f'}+\ketbra{y,k,f}{y,k,f'}\rb \\
&&-\dfrac{1}{2|S_{x,y}|} \sl_{f,f' \in S_{x,y}} \lb \ketbra{x,k,f}{y,k',f'}+\ketbra{y,k',f'}{x,k,f}\rb.
\label{eq:vjvj}
\eeea
Taking the sum of Eq.(\ref{eq:vjvj}) over $k,k' \in [n]$ yields
\beq
\Delta_{x,y}=A_{x,y} \otimes I_n + \dfrac{1}{n}B_{x,y} \otimes E_n,
\label{eq:deltaxy}
\eeq
where $I_n=\sum_{k \in [n]} \ketbra{k}{k}$ and $E_n=\sum_{k,k' \in [n]} \ketbra{k}{k'}$, and $A_{x,y}$ and $B_{x,y}$ are some matrices independent of $n$ (here we have switched the order of $k$-register and $f$-register). This holds for any $(x,y) \not\in \{(s,t),(t,s)\}$. 

On the other hand, we also have
\beq
\Delta_{s,t}+\Delta_{t,s}=\bar{A} \tp I_n + \dfrac{1}{n}\bar{B} \tp E_n,
\label{eq:deltastts}
\eeq
where $\bar{A}$ and $\bar{B}$ are some matrices independent of $n$. This can be derived from the fact that for any $j=(s,k,t,k')$ or $(t,k',s,k)$, where $(k,k') \neq (1,1)$, 
\beba
\ketbra{v_j}{v_j}
&=&\dfrac{1}{2|V_{T,l}|} \sl_{f,f' \in V_{T,l}} \lb \ketbra{s,k,f}{s,k,f'}+\ketbra{t,k,f}{t,k,f'}\rb \\
&&-\dfrac{1}{2|V_{T,l}|} \sl_{f,f' \in V_{T,l}} \lb \ketbra{s,k,f}{t,k',f'}+\ketbra{t,k',f'}{s,k,f}\rb
\eeea
and the fact that
\beba
\ketbra{v_{(s,1,t,1)}}{v_{(s,1,t,1)}}+\ketbra{v_{(t,1,s,1)}}{v_{(t,1,s,1)}}
&=&\dfrac{1}{|V_{T,l}|} \sl_{f,f' \in V_{T,l}} \lb \ketbra{s,1,f}{s,1,f'}+\ketbra{t,1,f}{t,1,f'}\rb \\
&&-\dfrac{1}{|V_{T,l}|} \sl_{f,f' \in V_{T,l}} \lb \ketbra{s,1,f}{t,1,f'}+\ketbra{t,1,f'}{s,1,f}\rb.
\eeea

Now by Eqs.(\ref{eq:deltaxy}) and (\ref{eq:deltastts}) we obtain
\beba
\Delta=A \otimes I_n+\dfrac{1}{n}B \otimes E_n,
\eeea
where $A$ and $B$ are some matrices independent of $n$. Then, some simple linear algebra shows that the spectrum of $\Delta$ (i.e. the set of eigenvalues sans multiplicity) does not depend on $n$. In particular, the smallest non-zero eigenvalue of $\Delta$ is $\omg{1}$, as desired.
\end{proof} 

\begin{lemma}
$\refl{A}$ can be implemented in $\polylog{n}$ time.
\label{lem:refla}
\end{lemma}
\begin{proof}
Since $\refl{A}$ is the reflection about the span of $\ket{i}\tp\ket{a_i}$'s, we can implement it as $U_AO_AU_A^{\dagger}$, where $U_A$ is any unitary operation that transforms $\ket{i}\tp\ket{\bar{0}}$ to $\ket{i}\tp \ket{a_i}$, and $O_A$ is the reflection about $\ket{\bar{0}}$ on the second subsystem. Obviously, $O_A$ can be implemented in $\polylog{n}$ time. So it remains to show that $U_A$ can be implemented in $\polylog{n}$ time. 

Observe that $\ket{a_i}$ can be written as:
\beq
\ket{a_{(x,k,f)}}=
\dfrac{1}{\sqrt{4n}}
\sl_{y \in M(x,f)} \sl_{k' \in [n]} Q\lb \ket{x,k,y,k'}+\ket{y,k',x,k}\rb,
\eeq
where $Q$ is a unitary operation acting on the span of $\ket{x,k,y,k'}$'s such that
\beba
Q\ket{x,k,y,k'}=\ket{x,k,y,k'},~~~~~~~\forall~(x,k,y,k)\not\in \{(s,1,t,1),(t,1,s,1)\};\\
Q\lb \dfrac{1}{\sqrt{2}}\ket{s,1,t,1}+\dfrac{1}{\sqrt{2}}\ket{t,1,s,1} \rb
=\dfrac{1}{\alpha}\ket{s,1,t,1}+\sqrt{1-\dfrac{1}{\alpha^2}}\ket{t,1,s,1}.
\eeea
So we can generate $\ket{i} \tp \ket{a_i}$ from $\ket{i}\tp \ket{\bar{0}}$ as follows:
\beba
\ket{x,k,f}\tp \ket{0,0,0,0}
&\to&
\ket{x,k,f}\tp \ket{x,k,0,0} \\
&\to&
\ket{x,k,f}\tp \dfrac{1}{\sqrt{n}} \sl_{k' \in [n]} \ket{x,k,0,k'} \\
&\to&
\ket{x,k,f}\tp \dfrac{1}{\sqrt{2n}} \sl_{y \in M(x,f)} \sl_{k' \in [n]} \ket{x,k,y,k'} \\
&\to&
\ket{x,k,f}\tp \dfrac{1}{\sqrt{4n}}  \sl_{y \in M(x,f)} \sl_{k' \in [n]} \lb \ket{x,k,y,k'}+\ket{y,k',x,k}\rb \\
&\to &
\ket{x,k,f}\tp \dfrac{1}{\sqrt{4n}}
 \sl_{y \in M(x,f)} \sl_{k' \in [n]} Q\lb \ket{x,k,y,k'}+\ket{y,k',x,k}\rb \\
&=&\ket{x,k,f} \tp \ket{a_{(x,k,f)}},
\eeea
where 
\bit
\item The first step is accomplished by performing a unitary operation $U_{A,1}$ that
transforms $\ket{x,k,0,0}$ to $\ket{x,k,x,k}$ on the first, second, fourth and fifth registers;
\item The second step is accomplished by performing a unitary operation $U_{A,2}$ that
transforms $\ket{0}$ to $\frac{1}{\sqrt{n}}\sum_{k' \in [n]}\ket{k'}$ on the last register;
\item The third step is accomplished by performing a unitary operation $U_{A,3}$ that
transforms $\ket{x,f} \tp \ket{0}$ to $\ket{x,f} \tp \frac{1}{\sqrt{2}}\sum_{y \in M(x,f)} \ket{y}$ on the first, third and sixth registers;
\item The fourth step is accomplished by performing a unitary operation $U_{A,4}$ that
transforms $\ket{x,k} \tp \ket{x,k,y,k'}$ to
$\ket{x,k} \tp \frac{1}{\sqrt{2}}(\ket{x,k,y,k'}+\ket{y,k',x,k})$ on all but the third registers,
\item The last step is accomplished by performing $Q$ on the last four registers.
\eit
Formally, let $U_A = Q U_{A,4}U_{A,3}U_{A,2}U_{A,1}$. Clearly, all of $U_{A,1}$, $U_{A,2}$, $U_{A,3}$, $U_{A,4}$ and $Q$ can be implemented in $\polylog{n}$ time, and hence so is $U_A$.
\end{proof}

\begin{lemma}
$\refl{B}$ can be implemented in $\polylog{n}$ time.
\label{lem:reflb}
\end{lemma}
\begin{proof}
Since $\refl{B}$ is the reflection about the span of $\ket{b_j}\tp\ket{a_j}$'s, we can implement it as $U_BO_BU_B^{\dagger}$, where $U_B$ is any unitary operation that transforms $\ket{\bar{0}} \tp \ket{j}$ to $\ket{b_j}\tp \ket{j}$, and $O_B$ is the reflection about $\ket{\bar{0}}$ on the first subsystem. Obviously, $O_B$ can be implemented in $\polylog{n}$ time. So it remains to show that $U_B$ can be implemented in $\polylog{n}$ time.

We generate $\ket{b_j}\tp\ket{j}$ from $\ket{\bar{0}}\tp\ket{j}$ as follows:
\beba
\ket{0,0,0}\tp \ket{x_1,k_1,x_2,k_2}
&\to &  \ket{0,0} \tp
\lb \dfrac{1}{\sqrt{|S_{x_1,x_2}|}}\sum\limits_{f \in S_{x_1,x_2}} \ket{f} \rb
\tp 
\ket{x_1,k_1,x_2,k_2}
\\
&\to & \dfrac{1}{\sqrt{2}}(\ket{x_1,k_1}-\ket{x_2,k_2})
\tp
\lb \dfrac{1}{\sqrt{|S_{x_1,x_2}|}}\sum\limits_{f \in S_{x_1,x_2}} \ket{f} \rb
\tp 
\ket{x_1,k_1,x_2,k_2}\\
&=&\ket{b_{(x_1,k_1,x_2,k_2)}} \tp \ket{x_1,k_1,x_2,k_2}.
\eeea
where
\bit
\item The first step is accomplished by performing a unitary operation $U_{B,1}$  that transforms $\ket{0} \tp \ket{x_1,x_2}$ to 
$\lb \frac{1}{\sqrt{|S_{x_1,x_2}|}} \sum_{f \in S_{x_1,x_2}} \ket{f} \rb \tp \ket{x_1,x_2}$ on the third, fourth and sixth registers.
\item The second step is accomplished by performing a unitary operation $U_{B,2}$  that transforms
$\ket{0,0} \tp \ket{x_1,k_1,x_2,k_2}$ to $\frac{1}{\sqrt{2}}(\ket{x_1,k_1}-\ket{x_2,k_2})\tp \ket{x_1,k_1,x_2,k_2}$  on the first two and last four registers.
\eit
Formally, let $U_B=U_{B,2}U_{B,1}$. Clearly, both $U_{B,1}$ and $U_{B,2}$ can be implemented in $\polylog{n}$ time, and hence so is $U_B$.
\end{proof}

Combining Lemma \ref{lem:gap}, Lemma \ref{lem:refla} and Lemma \ref{lem:reflb}, we know that $R_{\Lambda}$ can be implemented in $\polylog{n}$ time. Since $R_x$ can be also implemented in $\polylog{n}$ time, $U=R_{\Lambda}R_x$ can be implemented in $\polylog{n}$ time. Then, since $\P$ has witness size $\bgo{n}$, the time complexity of our algorithm is $\tbgo{n}$, as claimed.

\end{proof}

\section{Conclusion and open problems}
\label{sec:con}

To summarize, we have presented a time-efficient span-program-based quantum algorithm for the subgraph/not-a-minor problem for any tree. This algorithm has $\bgo{n}$ query complexity, which is optimal, and $\tbgo{n}$ time complexity, which is tight up to polylog factors.

Our work raises many interesting questions:
\bit
\item As mentioned in Section \ref{sec:sp}, our span program $\P$ accepts \emph{some} graphs that contain $T$ as a minor but not as a subgraph. Thus, it fails to solve the $T$-containment problem. However, it is not true that $\P$ accepts \emph{every} graph containing $T$ as a minor. One can see that a graph must possess certain structure in order to be accepted by $\P$. But this structure is not easy to characterize. It would be interesting to know exactly what kind of graphs are accepted by $\P$. If we have a better understanding of this matter, we might be able to modify $\P$ to make it reject any graph that does not contain $T$ as a subgraph, thus solving the $T$-containment problem.

\item In this paper, we have focused on the detection of trees. It is worth studying the detection of more complicated patterns, such as cycles and cliques. In particular, can we design time-efficient span-program-based quantum algorithms for those problems as well?

\item Perhaps the most interesting direction is to investigate the potential of our ``parallel-flow" technique for designing span programs.
In particular, can we use this technique to improve learning graphs \cite{lg}?
The basic idea of learning graph is to run a single flow from a vertex (i.e. the empty set) to some vertices (i.e. those contain a $1$-certificate) on an exponentially large graph, and its efficiency is determined by the energy of this flow. We observe that many decision problems can be decomposed into several \emph{correlated} subproblems, so that an input is a positive instance of the original problem if and only if it is a positive instance of all the subproblems simultaneously. Perhaps we can enhance the efficiency of learning graphs by dividing the original flow into several parallel flows, so that each flow corresponds to one subproblem. We might also need to make these flows interfere somehow in order to respect the correlations among these subproblems. It is possible that the total energy of these flows might be smaller than that of the original flow. Can we formalize this idea and use it to improve previous learning-graph-based quantum algorithms?

\eit

\section*{Acknowledgement}
We thank Ben Reichardt and Umesh Vazirani for helpful discussion on this topic. This research was supported by NSF Grant CCR-0905626 and ARO Grant W911NF-09-1-0440.

\begin{appendix}

\section{Proof of Claim \ref{clm:structure}}
\label{apd:pfclm}
For convenience, we introduce the following notation. Let $H=(V_H,E_H)$ be an arbitrary graph. For any $u \in V_H$, we say that $(H,u)$ is {\emph {good}} with respect to $T$ if there exists a subgraph $H'=(V_{H'},E_{H'})$ of $H$ and $\ecs \subseteq E_{H'}$ such that $\ecs(H')$ is isomorphic to $T$ and the root of $\ecs(H')$ contains $u$ (and thus $u$ must be involved in $\ecs$).

We will prove the following claim (which is stronger than Claim \ref{clm:structure}):
\begin{claim}
For any graph $G=(V_G,E_G)$ and coloring $c:V_G \to V_T$, there exist $L \in \N$, $W \subseteq c^{-1}(r)$, $\{V_{a,l} \subseteq c^{-1}(a): a \in V_{T,i}, l \in [L]\}$ and $\{V_{a,b,l} \subseteq c^{-1}(a): a \in V_{T,i}, b \in C(a), l \in [L]\}$, such that:
\ben
\item 
$\forall a \in V_{T,i}$, $\forall l \in [L]$, $V_{a,l}$ is the disjoint union of $V_{a,b,l}$'s for $b \in C(a)$; 
\item 
$c^{-1}(r)$ is the disjoint union of $W$ and $V_{r,l}$'s for $l \in [L]$;
\item 
$\forall a \in V_{T,i}$, $\forall b \in C(a) \cap V_{T,i}$, $\forall l \in [L]$, 
$N_{a \to b}(V_{a,b,l}) \subseteq V_{b,l}$
and
$N_{b \to a}(V_{b,l}) \subseteq V_{a,b,l}$;
\item
$\forall a \in V_{T,i}$, $\forall b \in C(a) \cap V_{T,l}$, $\forall l \in [L]$, $N_{a \to b}(V_{a,b,l})=\varnothing$;
\item 
Let $U_l = \cup_{a \in V_{T,i}} V_{a,l}$, $\forall l \in [L]$, and let $U = \cup_{l \in [L]} U_l$. Then $U_1,U_2,\dots,U_L$ are the vertex sets of the connected components of $G_c|_{U}$;
\item 
Let $Z = V_G \setminus U$. Then $\forall w \in W$, $(G_c|_Z,w)$ is good with respect to $T$.
\een
\label{clm:structurestrong}
\end{claim}

Let us first show that Claim \ref{clm:structurestrong} indeed implies Claim \ref{clm:structure}.
Suppose $G_c$ does not contain $T$ as a minor. By applying Claim \ref{clm:structurestrong}, we obtain the $W$, $V_{a,l}$'s and $V_{a,b,l}$'s satisfying the above conditions. Now define $V_a = \cup_{l \in [L]} V_{a,l}$ and $V_{a,b} = \cup_{l \in [L]} V_{a,b,l}$, for any $a \in V_{T,i}$ and $b \in C(a)$. We claim that these $V_a$'s and $V_{a,b}$'s satisfy all the conditions of Claim \ref{clm:structure}. This is because:
\bit
\item
Since the $V_{a,l}$'s and $V_{a,b,l}$'s satisfy conditions 1, 3 and 4 of Claim \ref{clm:structurestrong}, and the $V_a$'s or $V_{a,b}$'s are simply the union of the $V_{a,l}$'s or $V_{a,b,l}$'s respectively, it is obvious that the $V_a$'s and $V_{a,b}$'s satisfy
conditions 1, 3 and 4 of Claim \ref{clm:structure};
\item
Since $G_c$ does not contain $T$ as a minor, by condition 4 of Claim \ref{clm:structurestrong}, we must have $W=\varnothing$. Then, by condition 2 of Claim \ref{clm:structurestrong} and $V_r=\cup_{l \in [L]}V_{r,l}$, we have $V_r=c^{-1}(r)$. So condition 2 of Claim \ref{clm:structure} is also fulfilled.
\eit

Now it remains to prove Claim \ref{clm:structurestrong}. 
\begin{proof}[Proof of Claim \ref{clm:structurestrong}]
Proof by induction on the depth of $T$. 
\ben
\item \textbf{Basis}: Suppose $T$ has depth $0$. Namely, $T$ contains only a root node $r$. Then we simply let $W=c^{-1}(r)=V_G$ and $L=0$ (or let $V_{r,l}=\varnothing$ for any $l$). Then all the conditions of Claim \ref{clm:structurestrong} are trivially satisfied.

Suppose $T$ has depth $1$. Namely, $T$ contains a root node $r$ and its children $f_1,f_2,\dots,f_k$ (which are the leaves of $T$). We build the desired   
$W$, $V_{a,l}$'s and $V_{a,b,l}$'s as follows:
\bit
\item
Initially, set $W \la \varnothing$ and $L \la 0$. 
\item 
For each vertex $u \in c^{-1}(r)$, do: 
\ben
\item If $N_{r \to f_j}(u) \neq \varnothing$ for every $j \in [k]$, then set $W \leftarrow W \cup \{u\}$; 
\item Otherwise, there exists some $j \in [k]$ such that $N_{r \to f_{j}}(u)=\varnothing$. Then:
\bit
\item Set $L \la L+1$;
\item Set $V_{r,L} \la \{u\}$ and $V_{r,f_j,L} \la \{u\}$;
\item Set $V_{r,f_{j'},L} \la \varnothing$ for any $j' \neq j$.
\eit
\een
\eit

Now we show that the $W$, $V_{a,l}$'s and $V_{a,b,l}$'s obtained by this algorithm satisfy all the conditions of Claim \ref{clm:structurestrong}: 
\bit
\item
By the construction, it is obvious that $V_{r,l}$ is the disjoint union of
$V_{r,f_j,l}$'s for $j \in [k]$, for any $l$. So condition 1 is fulfilled;
\item
Since each vertex in $c^{-1}(r)$ is put into either $W$
or some $V_{r,l}$, we know that $c^{-1}(r)$ is the disjoint union of
$W$ and the $V_{r,l}$'s. So condition 2 is also fulfilled;
\item 
Since there is only one internal node $r$ and all of its children are leaves, the situation described by condition 3 does not exist and hence condition 3 is automatically satisfied;
\item  
By the definition of $V_{r,l}$'s and $V_{r,f_j,l}$'s, we have $N_{r \to f_{j}}(V_{r,f_j,l})=\varnothing$. So condition 4 is also satisfied;
\item 
Since each $U_l$ contains only one vertex, $G_c|_U$ is simply a collection of isolated vertices, and hence $U_1,U_2,\dots,U_L$ are the vertex sets of the connected components of $G_c|_U$. Therefore, condition 5 is fulfilled;
\item 
For any $w \in W$, it has a neighbor in each of $c^{-1}(f_1),c^{-1}(f_2),\dots,c^{-1}(f_k)$, and hence $G_c|_Z$ contains a tree that is isomorphic to $T$ and its root is $w$. So $(G_c|_Z, w)$ is good with respect to $T$. Thus, condition 6 is fulfilled.
\eit

\item \textbf{Inductive step}: Suppose that Claim \ref{clm:structurestrong} holds for any graph $G'$ and coloring $c'$ with respect to any tree $T'$ of depth at most $d$. 

Let $T$ be a tree of depth $d+1$. Suppose its root $r$ has $k$ children $d_1,d_2,\dots,d_k$. For each $j \in [k]$, since $T^{d_j}$ has depth at most $d$, we can apply Claim \ref{clm:structurestrong} to $G^{d_j}_{c}$ with respect to $T^{d_j}$, and obtain the $W^j$, $V^j_{a,\beta}$'s and $V^j_{a,b,\beta}$'s (where $a \in V^{d_j}_{T,i}$, $b \in C(a)$ and $\beta \in [L_j]$ for some $L_j$) satisfying the conditions of Claim \ref{clm:structurestrong}. In particular, let $U^j_{\beta}=\cup_{a \in V^{d_j}_{T,i}} V^j_{a,\beta}$ for $\beta \in [L_j]$, and let $U^j=\cup_{\beta \in [L_j]} U^j_\beta$, for $j \in [k]$. Then $U^j_1, U^j_2,\dots,U^j_{L_j}$ are the vertex sets of the connected components of $G_c|_{U^j}$ (which is a subgraph of $G^{d_j}_c$), for any $j \in [k]$. Let $Q^j = c^{-1}(r) \cup  U^j$ for $j \in [k]$, and let $Q = 
 \cup_{j \in [k]} Q^j $. Then let $H^j = G_c|_{Q^j}$, and let $H = G_c|_Q$. Note that each $H^j$ is a subgraph of $H$. 

Now consider the connected components of $H$ and $H^j$'s.
Suppose the vertex sets of the connected components of $H$ are $A_1, A_2, \dots,A_m$, and the vertex sets of the connected components of $H^j$ are $B^j_1,B^j_2,\dots,B^j_{m_j}$ for some $m_j$, for $j \in [k]$. Note that each $A_i$ is the union of some $B^j_t$'s (for different $(j,t)$'s), while each $B^j_t$ is the union of several $U^j_{\beta}$'s (for different $\beta$'s) and some subset of $c^{-1}(r)$. Let $E_{i,j} = \{t \in [m_j]: B^j_t \subseteq A_i\}$, and let $F_{j,t} = \{\beta \in [L_j]: U^j_{\beta} \subseteq B^j_t\}$, for $i \in [m]$, $j \in [k]$ and $t \in [m_j]$. Note that $A_i \cap c^{-1}(r) = \bigcup_{t \in E_{i,j}} \lb B^j_t \cap c^{-1}(r)\rb$ for any $i,j$.

We build the desired $W$, $V_{a,l}$'s and $V_{a,b,l}$'s as follows:
\bit
\item
Initially, set $W \la \varnothing$ and $L \la 0$. 
\item
For $i:=1$ to $m$ do:
\ben
\item If $N_{r \to d_j} (A_i \cap c^{-1}(r)) \cap W^j \neq \varnothing$ for every $j \in [k]$, then set $W \la W \cup (A_i \cap c^{-1}(r))$; 
\item Otherwise, there exists some $j \in [k]$ such that
$N_{r \to d_j} (A_i \cap c^{-1}(r)) \cap W^j=\varnothing$. Then:
\ben
\item[] For each $t \in E_{i,j}$ do:
\bit
\item Set $L \la L+1$;
\item Set $V_{r,L} \la B^j_t \cap c^{-1}(r)$ and 
$V_{r,d_j,L} \la B^j_t \cap c^{-1}(r)$;
\item Set $V_{a,L} \la \cup_{\beta \in F_{j,t}} V^j_{a,\beta}$ 
and
$V_{a,b,L} \la \cup_{\beta \in F_{j,t}} V^j_{a,b,\beta}$, for any $a \in V^{d_j}_{T,i}$ and $b \in C(a)$;
\item Set any other $V_{a,L}$ or $V_{a,b,L}$ to be $\varnothing$.
\eit
\een
\een
\eit

Now we show that the $W$, $V_{a,l}$'s and $V_{a,b,l}$'s obtained by the above algorithm satisfy all the conditions of Claim \ref{clm:structurestrong}.

\bit
\item
To prove that the $V_{a,l}$'s and $V_{a,b,l}$'s satisfy conditions 1, 3 and 4, we consider two possible cases separately:
\bit
\item $a \neq r:$
Since the $V^j_{a,\beta}$'s and $V^j_{a,b,\beta}$'s satisfy the conditions 1, 3 and 4 for each $j \in [k]$ (by the inductive hypothesis), and the $V_{a,l}$'s or $V_{a,b,l}$'s are simply the union of several $V^j_{a,\beta}$'s or $V^j_{a,b,\beta}$'s (for consistent choice of $j$'s and $\beta$'s) respectively, it is easy to see that the $V_{a,l}$'s and $V_{a,b,l}$'s also satisfy conditions 1, 3 and 4.
\item $a=r:$ 
\ben
\item By construction, for any $l$, exactly one of the $V_{r,d_j,l}$'s (for $j \in [k]$) equals $V_{r,l}$, and the other $V_{r,d_j,l}$'s are all empty.
So $V_{r,l}$ is indeed the disjoint union of the $V_{r,d_j,l}$'s (for $j \in [k]$). Hence, condition 1 is satisfied. 

\item For any internal node $d_j$, for any $l$, we have three possible cases:
\bit 
\item $V_{r,d_j,l}=V_{d_j,l}=\varnothing$; 
\item $V_{r,d_j,l}=B^j_t \cap c^{-1}(r)$ and $V_{d_j,l}=\cup_{\beta \in F_{j,t}} V^j_{d_j,\beta}$ for some $t \in [m_j]$;
\eit

For the first case, we have $N_{r \to d_j}(V_{r,d_j,l})=V_{d_j,l}=\varnothing$ and $N_{d_j \to r}(V_{d_j,l})=V_{r,d_j,l}=\varnothing$. For the second case, since $B^j_t$ is a connected components of $H^j$, we also have $N_{r \to d_j}(V_{r,d_j,l}) \subseteq V_{d_j,l}$ and $N_{d_j \to r}(V_{d_j,l}) \subseteq V_{r,d_j,l}$. 
Therefore, condition 3 is fulfilled.

\item For any leaf $d_j$, we have $W^j=c^{-1}(d_j)$. So if $N_{r \to d_j}(A_i \cap c^{-1}(r))\cap W^j=\varnothing$, then $N_{r \to d_j}(A_i \cap c^{-1}(r))=\varnothing$. This implies that
$N_{r \to d_j}(V_{r,d_j,l})=\varnothing$ for any $l$. Hence, condition 4 is also satisfied.
\een

\eit

\item
Since $c^{-1}(r) = \bigcup_{i \in [m]} (A_i \cap c^{-1}(r))$ and
$A_i \cap c^{-1}(r) = \bigcup_{t \in E_{i,j}} \lb B^j_t \cap c^{-1}(r)\rb$ for any $i,j$, by the construction of $W$ and the $V_{r,l}$'s, we know that they form a partition of $c^{-1}(r)$. Thus, condition 2 is also fulfilled.

\item
Let  $U_l=\cup_{a \in V_{T,i}} V_{a,l}$ for $l \in [L]$, and let $U=\cup_{l \in [L]} U_l$. Recall that $U^j_1,U^j_2.\dots,U^j_{L_j}$ are the vertex sets of the connected components of $G_c|_{U^j}$ for any $j \in [k]$ (by the inductive hypothesis). But viewing from the bigger graph $H^j$, different $U^j_{\beta}$ and $U^j_{\beta'}$ might become connected via some vertex in $c^{-1}(r)$. Now each $B^j_t$ is the vertex set of a connected component of $H^j$, so it is the union of some subset of $c^{-1}(r)$ and some $U^j_{\beta}$'s (for different $\beta$'s). Then, by construction,  we have:
\bit
\item
For any $l$, $U_l=B^j_t$ for some $j$ and $t$. So the vertices in $U_l$ are connected in $G_c|_U$;
\item
For any $l \neq l'$, $U_l$ and $U_{l'}$ are disjoint;
\item
For any $l \neq l'$, $U_l$ and $U_{l'}$ do not share any vertex in $c^{-1}(r)$. This is because:
\bit
\item For any $i \neq i'$, $A_i$ and $A_{i'}$ do not share any vertex in $c^{-1}(r)$, because $A_i$ and $A_j$ are the vertex sets of different connected components of $H$. It follows that, for any $i \neq i'$, $t \in E_{i,j}$ and $t' \in E_{i',j'}$, $B^j_t$ and $B^{j'}_{t'}$ do not share any vertex in $c^{-1}(r)$, since $B^j_t \subseteq A_i$ and $B^{j'}_{t'} \subseteq A_{i'}$;
\item Also, for any $t \neq t'$, $B^j_t$ and $B^j_{t'}$ do not share any vertex in $c^{-1}(r)$, since $B^j_t$ and $B^j_{t'}$ are the vertex sets of different connected components of $H^j$;
\item Now, for $l \neq l'$, we have $U_l \cap c^{-1}(r)=B^j_t \cap c^{-1}(r)$ for some $j$ and $t$, and $U_{l'} \cap c^{-1}(r)=B^{j'}_{t'} \cap c^{-1}(r)$ for some $j'$ and $t'$. There are two possible cases: (1) either $t \in E_{i,j}$ and $t' \in E_{i',j'}$ where $i \neq i'$; (2) or $t, t' \in E_{i,j}$ and $t \neq t'$. Either way, the above facts imply that $U_l$ and $U_{l'}$ do not share any vertex in $c^{-1}(r)$.
\eit
\item For any $(j,\beta) \neq (j', \beta')$, there is no edge between  $U^j_{\beta}$ and $U^{j'}_{\beta'}$ in $G_c|_U$;
\item The above facts imply that $U_l$ and $U_{l'}$ are disconnected in $G_c|_{U}$ for any $l \neq l'$.
\eit
Combining these facts, we know that $U_1,U_2,\dots,U_L$ are the vertex sets of the connected components of $G_c|_U$. So condition 5 is fulfilled.

\item
Let $Z^j = Y^{d_j}_c \setminus U^j$, for $j \in [k]$. Note that $Z^j \cap U=\varnothing$. Then, by the inductive hypothesis, for any $w \in W^j$, $(G_c|_{Z^j}, w)$ is good with respect to $T^{d_j}$. 

Now consider any $w \in W$. Since $c^{-1}(r) = \bigcup_{i \in [m]} (A_i \cap c^{-1}(r))$, there exists $i \in [m]$ such that $w \in A_i$. Then by construction, we must have $(A_i \cap c^{-1}(r)) \subseteq W$, and hence $A_i \cap V_{r,l} = \varnothing$ for any $l$, and hence $A_i \cap U=\varnothing$.

Pick any $w^j \in N_{r \to d_j}(A_i \cap c^{-1}(r)) \cap W^j$, for $j \in [k]$. Since $(G_c|_{Z^j}, w^j)$ is good with respect to $T^{d_j}$, $G_c|_{Z^j}$ contains a subgraph $\bar{G}_j$ such that $\bar{G}_j$ can be contracted into a tree $\bar{T}^j$ which is isomorphic to $T^{d_j}$ and the root of $\bar{T}^j$ contains $w^j$. Importantly, the vertices in $c^{-1}(r)$ and $U^j$ are not involved in such contractions (because they are not in $Z^j$). Do such contractions for each $j \in [k]$.

Meanwhile, since $A_i$ is a connected component of $H$, we can contract it into a single vertex. This contraction can be performed {\emph {simultaneously}} with the above contractions. The reason is that $A_i$ contains only some vertices in $c^{-1}(r)$ and $U^j$'s, which are not involved in any of the above contractions. Thus, even after the above contractions, we can still contract $A_i$ into a single vertex $v_i \deq A_i$ which contains $w$. Also, $v_i$ will be connected to $w^1,w^2,\dots,w^k$, which are the roots of $\bar{T}^1,\dots,\bar{T}^k$ which are isomorphic to $T^{d_1},T^{d_2},\dots,T^{d_k}$ respectively. Thus, the resulting graph contains a tree isomorphic to $T$.
Fig.\ref{fig:clm2transform} illustrates an example of such transformations.

Now since $A_i \cap U=\varnothing$ and $Z^j \cap U=\varnothing$ for any $j \in [k]$, the above transformation involves only some vertices in $Z=V_G \setminus U$. Hence, $(G_c|_Z,w)$ is good with respect to $T$. So condition 6 is also fulfilled.

\begin{figure}[htbp]
\center
\subfigure[]
{\includegraphics[scale=0.3]{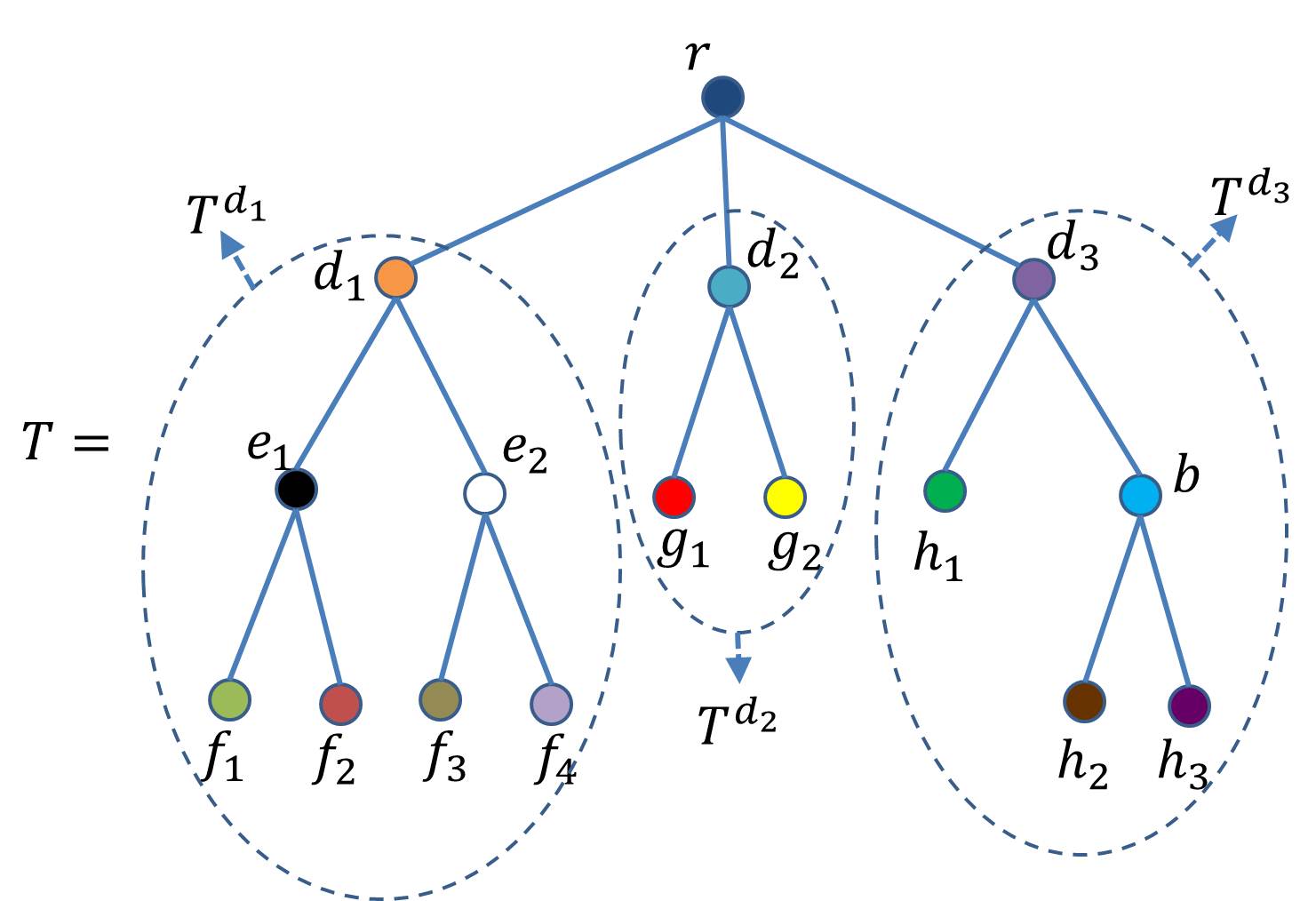}}
\subfigure[]
{\includegraphics[scale=0.4]{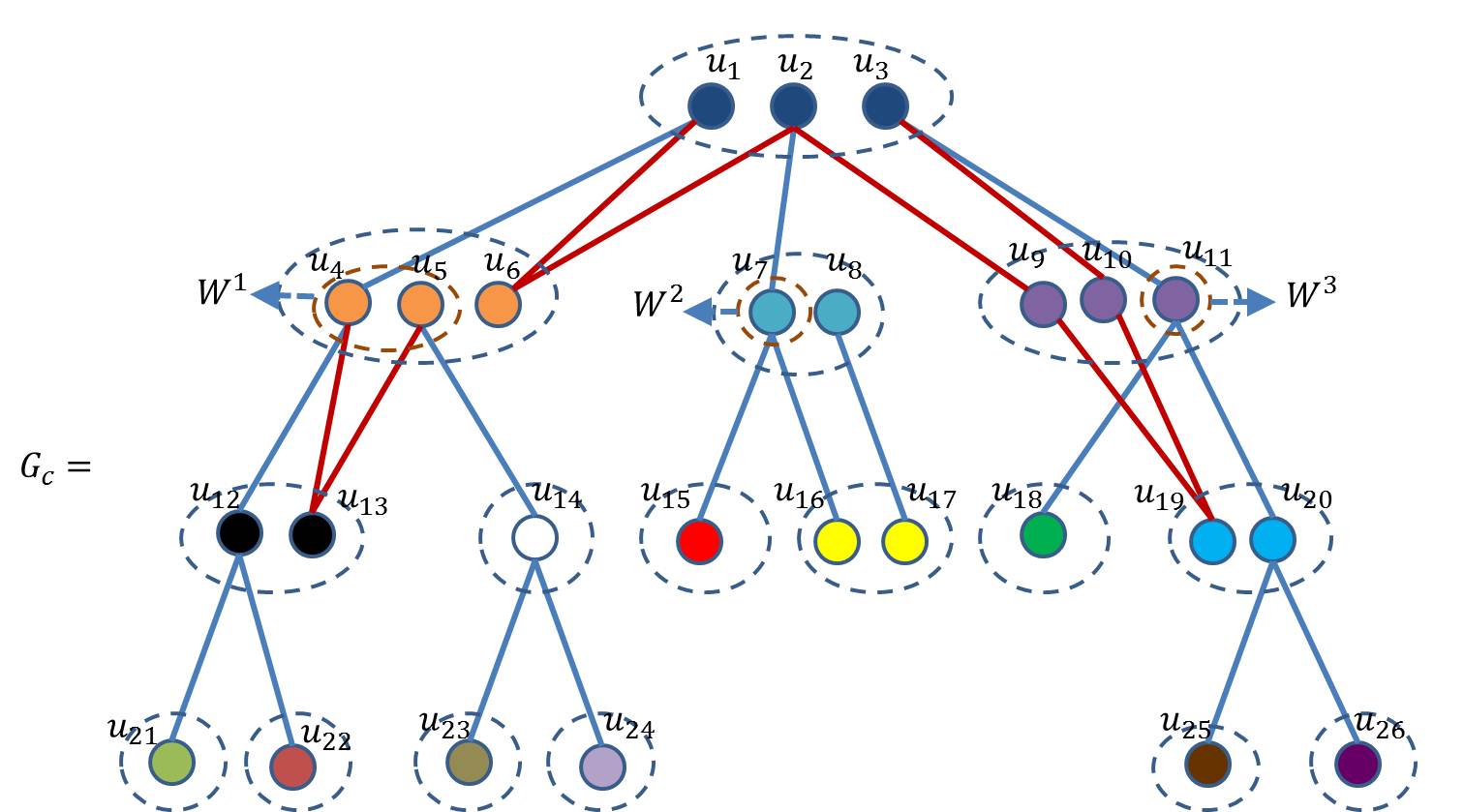}}
\subfigure[]
{\includegraphics[scale=0.4]{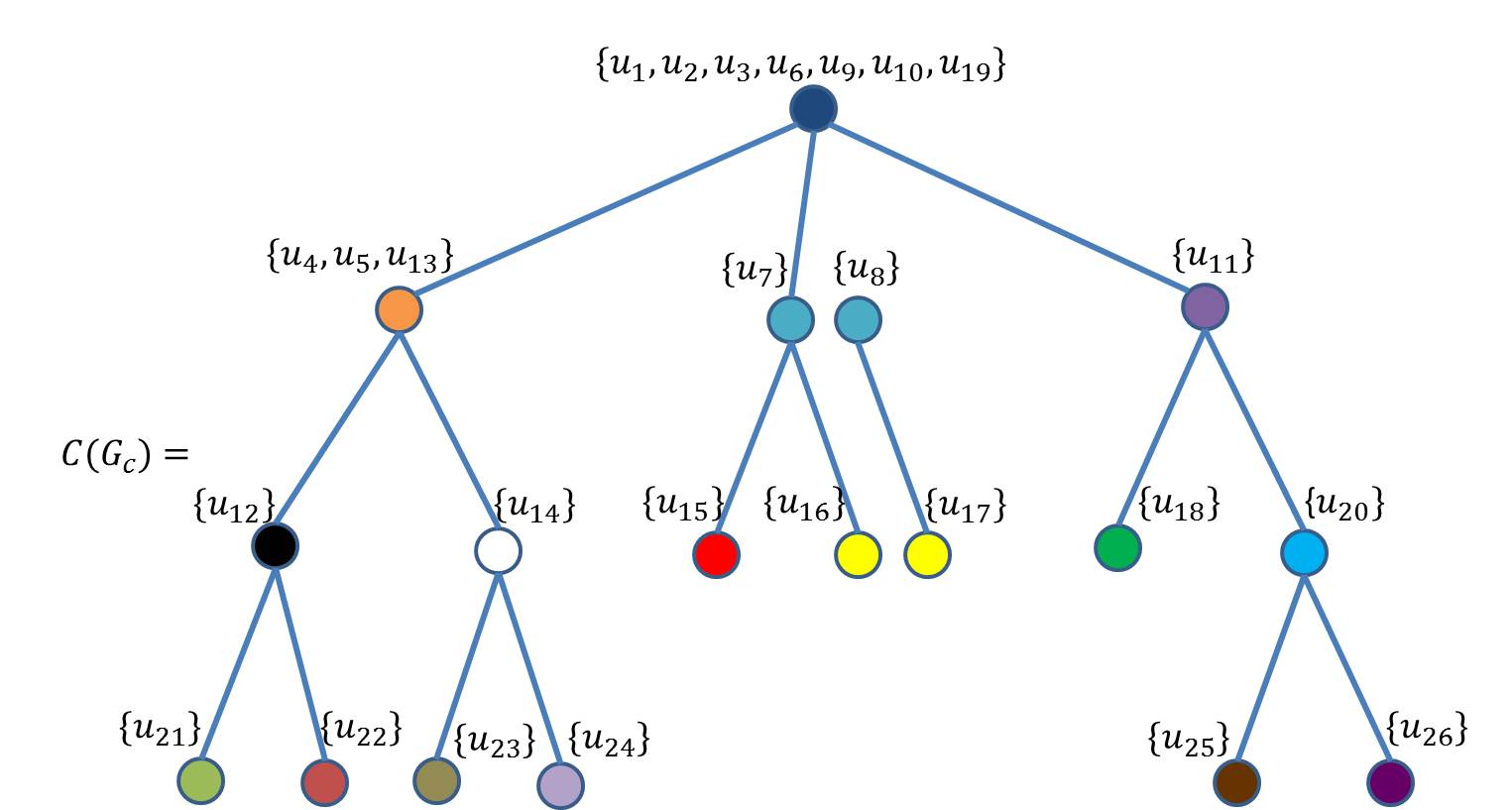}}
\caption{In this example, $T$'s root $r$ has $3$ children $d_1,d_2,d_3$.
For the given $G_c$, we have $W^1=\{u_4,u_5\}$, $U^1=\{u_6\}$,
$W^2=\{u_7\}$, $U^2=\{u^8\}$, $W^3=\{u_{11}\}$, $U^3=\{u_9,u_{10},u_{19}\}$,
$A_1=\{u_1,u_2,u_3,u_6,u_9,u_{10},u_{19}\}$ and $A_2=\{u_8\}$.
Note that $N_{r \to d_1}(A_1 \cap c^{-1}(r)) \cap W^1=\{u_4\}$,
$N_{r \to d_2}(A_1 \cap c^{-1}(r)) \cap W^2=\{u_7\}$ and $N_{r \to d_3}(A_1 \cap c^{-1}(r)) \cap W^3=\{u_{11}\}$.
By contracting the edges $(u_1,u_6)$, $(u_2,u_6)$, $(u_2,u_9)$, $(u_9,u_{19})$, $(u_{10},u_{19})$, $(u_3,u_{10})$, all the vertices in $A_1$ are combined together. We also contract the edges $(u_4,u_{13})$ and $(u_5,u_{13})$ to obtain a tree isomorphic to $T^{d_1}$ in the subgraph $G_c|_{Z^1}$. Let $\ecs$ be the set of the aforementioned edges.
Then, the resulting graph $\ecs(G_c)$ contains a tree isomorphic to $T$. Hence, $W=\{u_1,u_2,u_3\}$.}
\label{fig:clm2transform}
\end{figure}

\eit
\een

\end{proof}

\end{appendix}

\end{document}